\documentclass[a4paper]{article}

\usepackage{arxiv}

\usepackage[utf8]{inputenc} 
\usepackage[T1]{fontenc}    
\usepackage[final]{graphicx}
\usepackage{hyperref}       
\usepackage{url}            
\usepackage{booktabs}       

\usepackage{microtype}      

\usepackage{enumitem}       

\usepackage{calc,ifthen}

\usepackage{changes}
\definechangesauthor[color=red]{GK}
\definechangesauthor[color=green]{AG}

\usepackage{amsmath,amsthm,amssymb}
\usepackage{amsfonts}       
\usepackage{nicefrac}       

\usepackage[textsize=normalsize,textwidth=2cm,colorinlistoftodos]{todonotes}

\usepackage{tikz,pgfplots}
\usetikzlibrary{arrows.meta,decorations.pathmorphing,
     decorations.markings,shadows,calc}

\usepackage[caption = false]{subfig}

\pgfplotsset{compat=newest,
      width=0.9\textwidth,height=0.9\textwidth/1.618,
      every tick/.append style={black,line width=1pt},
      every axis/.append style={line width=1pt},
      enlargelimits=false,
      axis lines=middle,
      every inner x axis line/.append style={->},
      every inner y axis line/.append style={->},
      every axis y label/.style={at={(0,1)},above right},
      every axis x label/.style={at={(1,0)},above right},
      every pin edge/.style={solid,black}
}

\newcommand{\referencelist}{
	\bibliographystyle{plain}
	\bibliography{references,ArtLibrary}
}

\newtheorem{theorem}{Theorem}[section]

\newcommand{\ensem}[1]{\langle #1 \rangle}

\newcommand{\ie}{\textit{i.e.}\/, }
\newcommand{\eg}{\textit{e.g.}\/, }

\providecommand*{\diff}{\operatorname{d}\!}

\providecommand*{\mrm}[1]{\mathrm{#1}}
\providecommand*{\eu}{\ensuremath{\mrm{e}}}
\providecommand*{\iu}{\ensuremath{\mrm{i}}}

\providecommand{\renewoperator}[3]{%
       \renewcommand*{#1}{\mathop{#2}#3}}

\renewoperator{\Re}%
               {\mathrm{Re}}{\nolimits}
\renewoperator{\Im}%
               {\mathrm{Im}}{\nolimits}

\newif\ifgreek
\def\testgreek#1{
  \ifx#1\alpha\greektrue\else
  \ifx#1\beta\greektrue\else
  \ifx#1\gamma\greektrue\else\ifx#1\Gamma\greektrue\else
  \ifx#1\delta\greektrue\else\ifx#1\Delta\greektrue\else
  \ifx#1\epsilon\greektrue\else
  \ifx#1\zeta\greektrue\else
  \ifx#1\eta\greektrue\else
  \ifx#1\theta\greektrue\else\ifx#1\Theta\greektrue\else
  \ifx#1\iota\greektrue\else
  \ifx#1\kappa\greektrue\else
  \ifx#1\lambda\greektrue\else\ifx#1\Lambda\greektrue\else
  \ifx#1\mu\greektrue\else
  \ifx#1\nu\greektrue\else
  \ifx#1\xi\greektrue\else\ifx#1\Xi\greektrue\else
  \ifx#1\pi\greektrue\else\ifx#1\Pi\greektrue\else
  \ifx#1\rho\greektrue\else
  \ifx#1\sigma\greektrue\else\ifx#1\Sigma\greektrue\else
  \ifx#1\tau\greektrue\else
  \ifx#1\upsilon\greektrue\else\ifx#1\Upsilon\greektrue\else
  \ifx#1\phi\greektrue\else\ifx#1\Phi\greektrue\else
  \ifx#1\chi\greektrue\else
  \ifx#1\psi\greektrue\else\ifx#1\Psi\greektrue\else
  \ifx#1\omega\greektrue\else\ifx#1\Omega\greektrue\else
  \ifx#1\varepsilon\greektrue\else
  \ifx#1\vartheta\greektrue\else
  \ifx#1\varrho\greektrue\else
  \ifx#1\varsigma\greektrue\else
  \ifx#1\varphi\greektrue\else
     \greekfalse
  \fi\fi\fi\fi\fi\fi\fi\fi\fi\fi
  \fi\fi\fi\fi\fi\fi\fi\fi\fi\fi
  \fi\fi\fi\fi\fi\fi\fi\fi\fi\fi
  \fi\fi\fi\fi\fi\fi\fi\fi\fi}

\renewcommand{\vec}[1]{\boldsymbol#1}
\newcommand{\unitvec}[1]{\hat{\vec{#1}}}
\newcommand{\R}{\mathbb{R}}
\newcommand{\reg}{\mathcal R} 
\newcommand{\numdensity}{{\mathfrak n}} 

\newcommand{\dv}{\vec{d}}
\newcommand{\rv}{\vec{r}}

\newcommand{\kvh}{\unitvec{k}}
\newcommand{\kvhi}{\kvh_{}}
\newcommand{\rvh}{\unitvec{r}}
\newcommand{\nuvh}{\unitvec{\nu}}
\newcommand{\xvh}{\unitvec{x}}
\newcommand{\yvh}{\unitvec{y}}
\newcommand{\zvh}{\unitvec{z}}

\newcommand{\ui}{u_{\mathrm{in}}}
\newcommand{\us}{u_{\mathrm{sc}}}

\newcommand{\Vi}{V_{\mathrm{in}}}

\numberwithin{equation}{section}

\title{Effective Waves for Random Three-dimensional Particulate Materials}

\author{
  Artur L.~Gower\thanks{Webpage: \href{https://arturgower.github.io}{arturgower.github.io}} \\
  Department of Mechanical Engineering\\
  University of Sheffield\\
  Sheffield, UK \\
  \texttt{arturgower@gmail.com} \\
   \And
 Gerhard Kristensson\thanks{Webpage: \href{https://www.eit.lth.se/personal/gerhard.kristensson}{www.eit.lth.se/personal/gerhard.kristensson}}\\
  Department of Electrical and Information Technology\\
  Lund University\\
  P.O. Box 118\\
  SE-221 00 Lund, Sweden\\
  \texttt{gerhard.kristensson@eit.lth.se}\\
  }

 \usepackage[most]{tcolorbox} 

\newtcolorbox[auto counter]{optionalnote}[2][]{
    parbox=false,
    colbacktitle= white,
    colback=green!5!white,
    colframe=white!45!black,
    coltitle=black,
    enhanced,
    attach boxed title to top left={yshift=-2mm},
    title={\thetcbcounter.~#2}
,#1}

\newtcolorbox{highlight-result}[1][]{
 parbox=false,
 boxrule=0pt,top=0pt,bottom=0pt,
colback=blue!8!white,
enhanced,#1}

\tcbset{colback=blue!8!white, colframe=white!50!white}

\begin{document}
\maketitle

\setcounter{footnote}{0}

\begin{abstract}
How do you take a reliable measurement of a material whose microstructure is random? When using wave scattering, the answer is often to take an ensemble average (average over time or space). By ensemble averaging we can calculate the average scattered wave and the effective wavenumber. To date, the literature has focused on calculating the effective wavenumber for a plate filled with particles. One clear unanswered question was how to extend this approach to a material of any geometry and for any source. For example, does the effective wavenumber depend on only the microstructure, or also on the material geometry? In this work, we demonstrate that the effective wavenumbers depend on only microstructure and not the geometry, though beyond the long wavelength limit there are multiple effective wavenumbers. We show how to calculate the average wave scattered from a random particulate material of any shape, and for broad frequency ranges. As an example, we show how to calculate the average wave scattered from a sphere filled with particles.
\end{abstract}


\section{Introduction}
Under close inspection, many natural and synthetic materials are composed of small randomly distributed particles. This is why techniques to measure and predict these particle properties are important in many areas of science and engineering. Waves, either mechanical (like sound) or electromagnetic, are an excellent choice to probe particles because they can be non-invasive and energy efficient.
\begin{figure}[ht]
    \centering
    \includegraphics[width=0.49\linewidth]{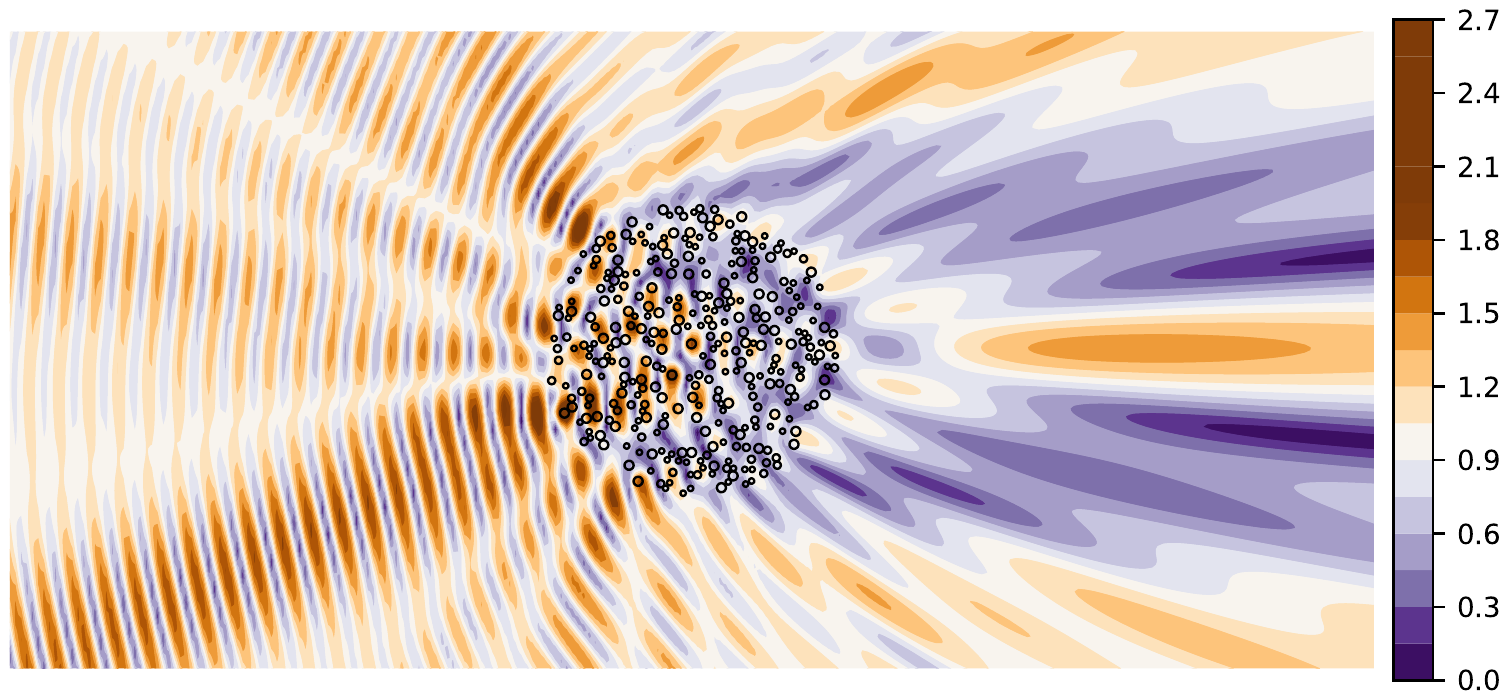}
       \includegraphics[width=0.49\linewidth]{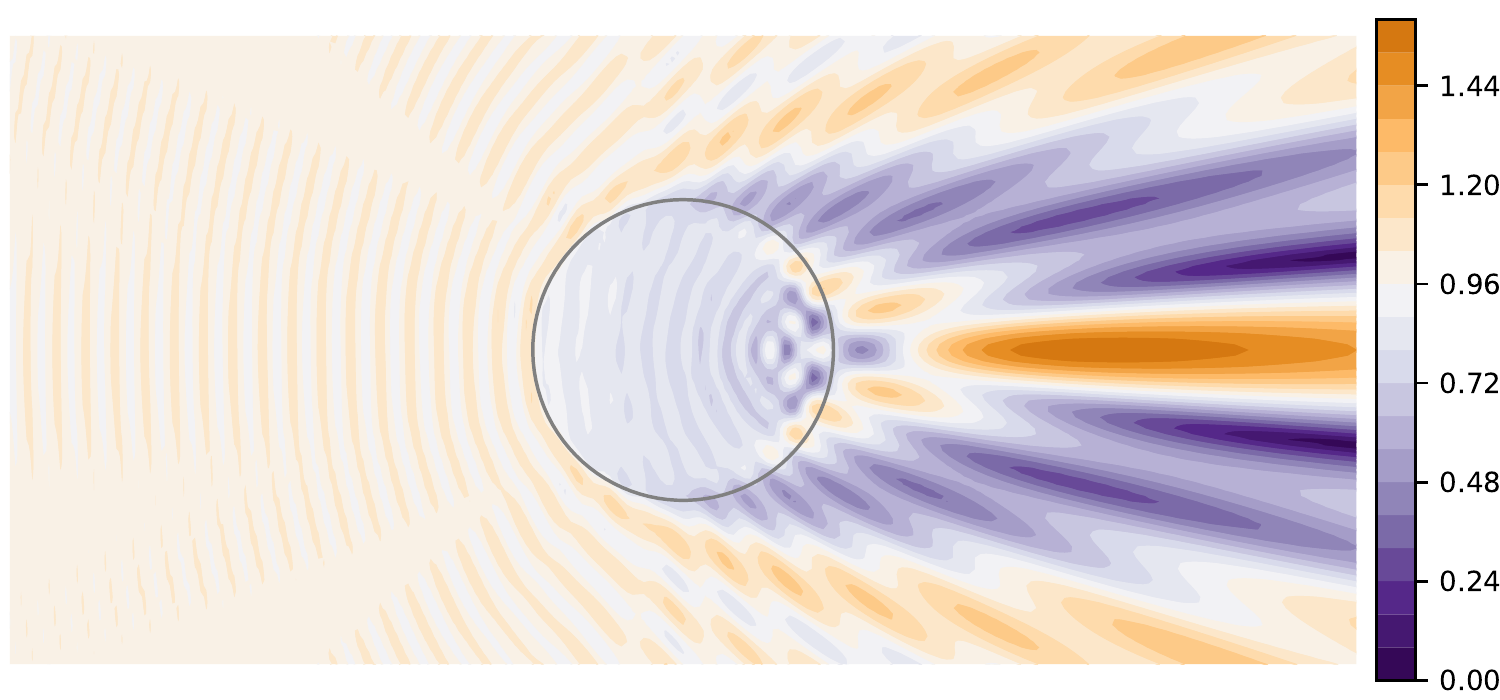}
    \caption{The figure on the left illustrates the scattered field from one configuration of particles due to an incident plane-wave. The colour indicates the field value. The figure on the right illustrates the ensemble averaging of the scattered field over every possible particle configuration. Although the figure on the left shows the scattered field for one moment in time, the figure on the right is what many sensors would measure when averaging over time or space. }
    \label{fig:compare-one-configuration}
\end{figure}

{\bf Sensing application.} To develop non-invasive sensors we first need efficient mathematical models on how waves scatter. Particulate materials are valuable products across many industries. They are present in pharmaceuticals (powders and emulsions) and aerosols (suspension); metal and polymers powders for additive manufacturing, and many chemical and food industries\footnote{Malvern Panalytical: \url{www.malvernpanalytical.com/en/industries}. Horiba Scientific: \url{www.horiba.com/en_en/products/by-segment/scientific/particle-characterization/applications/}. } (emulsions, colloids, slurry). Accurately monitoring the particles during processing (chemical, thermal, or mechanical) can enable automation and lead to optimised particle properties. Optimising particulates reduces waste and improves product quality. Currently there is no quantitative method to monitor dense particulates. The most reliable methods currently rely on light diffraction, which is only viable if the material is mostly transparent. In practice, this means the material needs to be diluted or filtered before applying these diffraction methods, which can only be done in small batches using in a controlled laboratory environment.

To develop new sensors to monitor in real time, large quantities of materials, we first need to understand how waves scatter from these dense particulates, and develop efficient models to describe this scattering.

\paragraph{Every particle counts.} Any method that uses waves to probe a particulate material needs to consider how each particle scatters waves.
This is because both particle properties and positions influence the total scattered waves, as shown in Figure~\ref{fig:compare-one-configuration}.

Although it is possible to numerically simulate scattered waves from a specific arrangement of particles, these numerical methods are computationally too intensive for most practical applications. For example, one droplet of most emulsions will contain hundreds of millions of oil particles, whose positions are unknown. The most successful methods avoid these heavy computations by replacing the material with an equivalent homogeneous material\cite{challis_ultrasound_2005,mishchenko_first-principles_2016}. This equivalent homogeneous material is calculated by taking an \emph{ensemble average}.

\paragraph{The ensemble average.} Ensemble averaging not alone simplifies the calculations, it is also the route to devising measurements which do not depend on the positions of the particles, which are unknown. One way to do this is to take the average of the scattered field. This average can be taken over space or over time (for ergodic systems). Both of these types of average measurements eliminate the need to know the particle position, and so lead to reliable measurements\cite{foldy_multiple_1945,mishchenko_first-principles_2016}.

If $u(x)$ represents the transmitted wave field, measured at some distance $x$, then, for a plane wave source propagating along the $x$-axis, it is common to approximate the ensemble average as a plane wave of the form
\begin{equation} \label{eqn:one-transmitted}
\ensem{u(x)} \approx A \eu^{\iu \omega (x/c_* - t) - \alpha x},
\end{equation}
where $c_*$ is the (effective) wave speed, $\alpha$ the rate of attenuation, and $A$ the average transmission coefficient.
The process of calculating the ensemble average links the measurables $c_*$, $\alpha$, and $A$ to the particles; it is this link which drives many sensing methods.
It is common to combine $c_*$ and $\alpha$ into one quantity, the complex effective wavenumber: $k_* = \omega /c_* + \iu \alpha$.

\paragraph{What is known.} One scenarios has been mostly clearly understood: a plane wave incident on a halfspace or plate region filled with particles. This setup has, what we call in this paper, {planar} symmetry. For {planar} symmetry, in the limits of low frequency or low volume fraction, there are explicit formulas~\cite{parnell_multiple_2010,linton_multiple_2006, martin_multiple_2006, linton_multiple_2005,caleap_effective_2012,caleap2015metamaterials}, and an understanding on how to calculate wave reflection and transmission~\cite{martin_multiple_2011,gower2019multiple}. Further, the effective wavenumbers for planar symmetry have also been rigorously deduced~\cite{gower2019proof} (given typical statistical assumptions), though there is often more than one effective wavenumber for the same fixed frequency~\cite{gower2019multiple,willis2020transmission,willis2020transmission-meta}. One clear question that remained was how to extend this approach to a material with any geometry and for any source? For example, is the effective wavenumber $k_*$ the same for other geometries?  There has even been evidence~\cite{guerin2006effective} that the effective properties (and wavenumber) depend on the geometry of the material. 
If this were true, these effective wavenumbers would not be very useful, as they would change for every sample of the same material.

In the electromagnetic community, the analysis of effective wave properties in particulate media has a long tradition. Some of the most significant contributions are collected in textbooks, \eg~\cite{Tsang+Kong2001,Tsang+etal2001,Tsang+Kong+Ding2000} and journal literature~\cite{mishchenko_first-principles_2016,Tishkovets+etal2011}. With a few exceptions, the analysis deals again with planar symmetry. 

\paragraph{This paper.} Here we develop the theory for effective waves and wavenumbers for materials in any geometry. The key to achieve this is to use the representation:
\begin{equation}
\ensem{u(x)} = \sum_{p=1}^P \phi_p(x) \eu^{- \iu \omega t},
\end{equation}
where $\phi_p(x)$ is a function that satisfies $\nabla^2 \phi_p(x) + k_p^2 \phi_p(x) = 0$.
This representation allows us to deduce a dispersion equation for the $k_p$ that does not depend on the material geometry. This question of whether the geometry changes the effective wavenumbers has been raised in previous studies~\cite{guerin2006effective}.

In this paper we present a framework for effective scalar waves in any material geometry, and then specialise to a material shaped as a sphere and a plate. This allows us to design highly efficient numerical methods for these cases.

\section{A collection of particles}
We begin with the deterministic many-particle scattering problem and use the Null-field approach~\cite{Kristensson2016}.
Consider $N$ different particles, where the $i$-th particle is centred at the location $\rv_i$
as shown in Figure~\ref{fig:CollectionScatterer}.\footnote{Throughout this paper, vector-valued quantities are denoted in italic boldface and vectors of unit length have a ``hat'' or caret ($\unitvec{\;}$) over the symbol.}
The radius of the minimum circumscribed sphere, centred at $\rv_i$, is $a_i$, $i=1,2,\ldots,N$. We assume that no minimum circumscribed spheres intersect. Each particle can have a different shape and material properties.

The particles are located in a  homogeneous, isotropic media with wavenumber $k$, which is either a real number or a complex number with a positive imaginary part.

The prescribed sources are located in the region $\Vi$, which is a region disjoint to all particles,\footnote{More precisely, the circumscribed sphere of the source region must not include any local origin $\rv_i$, $i=1,2,\ldots,N$. For instance, an incident plane wave fulfils these restrictions.} and these sources generate the field $\ui(\rv)$ everywhere outside $\Vi$.

For a point $\rv$, outside of the circumscribed spheres of all particles, we can write the total field $u(\rv)$ as a sum of the incident wave $\ui(\rv)$ and all scattered waves in the form~\cite{Kristensson2015a,Kristensson2016,Linton+Martin2006}
\begin{equation}
    u(\rv) = \ui(\rv) + \us(\rv), \quad \us(\rv) =  \sum_{i=1}^N \sum_n f_n^i \mathrm u_n (k \rv - k \rv_i),
    \label{eqn:total_discrete_wave}
\end{equation}
where we assumed $ |\rv - \rv_i| > a_i $ for $i=1,2,\ldots N$, the $f_n^i$ are coefficients we need to determine, and for convenience we use scalar spherical waves:
\begin{equation}
\left\{\begin{aligned}
    & \mathrm u_{n}(k\rv)={\mathrm{h}}_\ell^{(1)}(kr)\mathrm{Y}_{n}(\rvh), & \text{(outgoing spherical waves)}
    \label{eqn:outgoing_waves_and_regular_waves}
    \\
    & \mathrm v_{n}(k\rv)=\mathrm{j}_\ell(kr)\mathrm{Y}_{n}(\rvh), & \text{(regular spherical waves)}
 \end{aligned}\right.
\end{equation}
where $r=|\rv|$, and $n$ denotes a multi index $n=\{\ell,m\}$, with summation being over $\ell=0,1,2,3\ldots$ and $m=-\ell,-\ell+1,\ldots,-1,0,1,\ldots,\ell$.
For more details, see Appendix~\ref{sec:Functions}.
The spherical Hankel and Bessel functions are denoted ${\mathrm{h}}_\ell^{(1)}(z)$ and $\mathrm{j}_\ell(z)$, respectively.
The field $\sum_n f_n^i \mathrm u_n(k \rv - k \rv_i)$ is the wave scattered from particle-$i$.

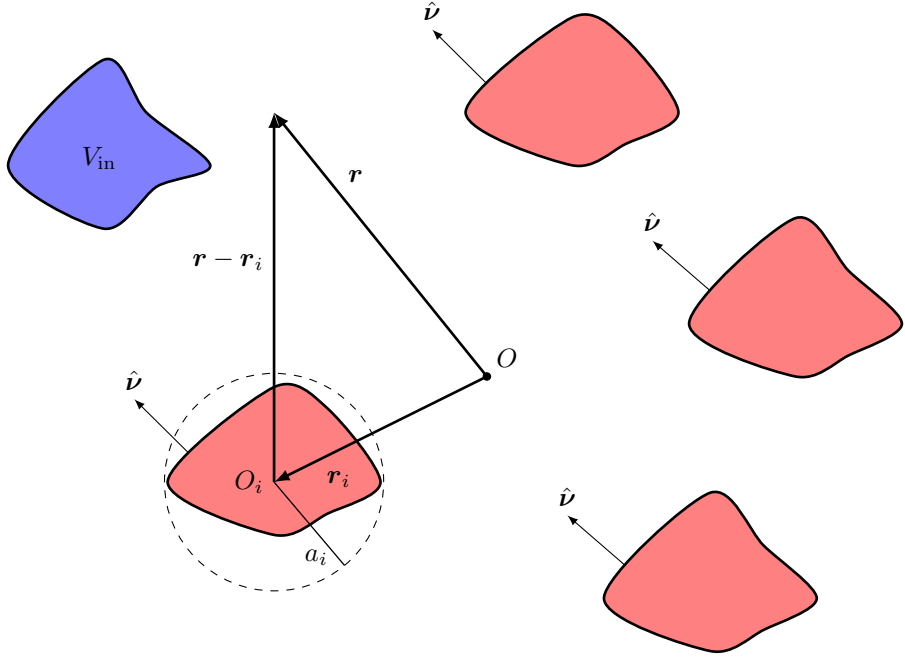
\begin{figure}[t]
    \begin{center}
\begin{tikzpicture}[>=latex,scale=1.4,decoration={markings,mark=at position 0.6 with {\draw[black,->,thin] (0,0) -- (0,-1)node[above]{$\unitvec{\nu}$};}}]
    \def\rmax{1.03cm}
    \def\domain{(0,-0.5) (0.5,-0.3) (1,0) (0.5,0.5) (0,1) (-1,0)}
    \def\domainalt{(0,-0.5) (0.5,-0.3) (1,0) (0.5,0.7) (0,0.9) (-1,0)}
    \coordinate (O) at (0,0);
    \coordinate (O1) at (2.9,0.5);
    \coordinate (O2) at (0.8,2.5);
    \coordinate (Oi) at (-2,-1);
    \coordinate (ON) at (2.1,-2.1);
    \coordinate (r) at (-2,2.5);
    \filldraw[line width=1pt,smooth cycle,fill=blue!50,draw=black,shift={(-3.6,2)}] plot coordinates {(0,-0.6) (0.5,-0.2) (1,0) (0.4,0.5) (0,1) (-0.9,0)}
     node[xshift=1.2cm,yshift=0.1cm] {$\Vi$};
    \filldraw[postaction={decorate},line width=1pt,smooth cycle,fill=red!50,draw=black,shift={(O1)}] plot coordinates \domain;
    \filldraw[postaction={decorate},line width=1pt,smooth cycle,fill=red!50,draw=black,shift={(O2)}] plot coordinates \domainalt;
    \filldraw[postaction={decorate},line width=1pt,smooth cycle,fill=red!50,draw=black,shift={(Oi)}] plot coordinates \domainalt;
    \filldraw[postaction={decorate},line width=1pt,smooth cycle,fill=red!50,draw=black,shift={(ON)}] plot coordinates \domain;
    \draw[fill=black] (O) circle (1pt) node[above right] {$O$};
    \draw[line width=1pt,->] (O)--(Oi) node[left] {$O_i$} node[pos=0.8,anchor=north west] {$\rv_i$};
    \draw[line width=1pt,->] (Oi)--(r) node[pos=0.6,anchor=east] {$\rv-\rv_i$};
    \draw[line width=1pt,->] (O)--(r) node[pos=0.7,anchor=south west] {$\rv$};
    \draw[dashed] (Oi) circle (\rmax);
    \draw[line width=0.5pt] (Oi)--+(-50:\rmax) node[pos=0.9,left] {$a_i$};
\end{tikzpicture}
    \end{center}
    \caption{The geometry of a collection of  the $N$ particles and the region of prescribed sources $\Vi$.
    The common origin is denoted $O$.
    The positions of the local origins $O_i$ are $\rv_i$, $i=1,\ldots,N$, and the radius of the minimum circumscribed sphere of each local particle is $a_i$.}
    \label{fig:CollectionScatterer}
\end{figure}

\subsection{Incident field}
We assume the incident field is generated outside of all particles, see Figure~\ref{fig:CollectionScatterer}, so it has an expansion in regular spherical waves
\begin{equation}
    \ui(\rv)=\sum_{n}g_n \mathrm v_n(k\rv)=\sum_{nn'}g_n\mathcal{V}_{nn'}(k\rv_i)\mathrm v_{n'}(k\rv - k\rv_i),
    \label{eqn:incident-spherical}
\end{equation}
where for the last equality we used a translation matrix of the regular spherical waves, $\mathcal{V}_{nn'}(k\rv_i)$, to write the incident wave in terms of spherical waves centred at $\rv_i$. See Appendix~\ref{sec:Translation} for details.

In many applications, we adopt a plane wave impinging along the direction $\kvhi$, \ie
\begin{equation}
\ui(\rv)=\eu^{\iu k\kvhi\cdot\rv},
\end{equation}
where the expansion coefficients, $g_{n}$, are given by~\cite{Kristensson2016}
\begin{equation}\label{eq:a_n def}
g_{n}=4\pi\iu^{\ell}\mathrm{Y}_{n}^*(\kvhi),
\end{equation}
where a star ${}^*$ denotes complex conjugate.
For the special direction $\kvhi=\zvh$, the coefficients are simplified
\begin{equation}
g_{n}=4\pi\iu^\ell\mathrm{Y}_{n}^*(\zvh)=\iu^\ell\delta_{m,0}\sqrt{4\pi(2\ell+1)}.
\end{equation}

\subsection{Scattered field}
The coefficients $f_n^i$ from~\eqref{eqn:total_discrete_wave} are determined by using the T-matrix to relate the field incident on the $i$-th particle, $u(\rv) - \sum_n f_n^i \mathrm u_n(k \rv - k \rv_i)$, to the wave scattered from the $i$-th particle, $\sum_n f_n^i \mathrm u_n(k \rv - k \rv_i)$, which leads to~\cite{Kristensson2015a,Kristensson2016,Linton+Martin2006}
\begin{equation}\label{eq:Individual-expansion-coefficients}
   \tcboxmath{ f_n^i=\sum_{n'} T_{n}^i\mathcal{V}_{n'n}(k\rv_i)g_{n'}
    +\sum_{\substack{j=1\\j\neq i}}^N\sum_{n'}T_{n}^i\mathcal{ U}_{n'n}(k\rv_i - k\rv_j)f_{n'}^j,\quad i=1,2,\ldots,N,}
\end{equation}
where $\mathcal{U}_{nn'}$ is the translation matrix of the outgoing spherical waves $\mathrm u_n$, see Appendix~\ref{sec:Translation}. Above we have used a diagonal T-matrix which assumes a spherical particle; later we will explain how this leads to the solution for non-spherical particles whose orientation is independent of position and properties.

Equation~\eqref{eq:Individual-expansion-coefficients} is very difficult to calculate when the number of particles $N$ is large.
Nevertheless, there are several software packages making substantial progress, \eg MSTM (Multiple Sphere T Matrix)~\cite{Mackowski+Mishchenko2011a,ganesh_far-field_2010,ganesh_algorithm_2017}.

The  $T_{n}^i$ depend only on the properties of the $i$-th particle, while the scattering coefficients $f_n^i$ depends on the positions and properties of all the particles. For example, for acoustics, and a homogeneous spherical particles, we would have~\cite{linton_multiple_2006}:
\begin{equation} \label{eqn:T-matrix-acoustics}
    T^i_{n} = - \frac{\gamma_i \mathrm j_\ell (k a_i) \mathrm j_\ell (k_i a_i) -   \mathrm j_\ell (k a_i) \mathrm j_\ell (k_i a_i)}{\gamma_i \mathrm h^{(1)\prime}_\ell (k a_i) \mathrm j_\ell (k_i a_i) -   \mathrm h^{(1)}_\ell (k a_i) \mathrm j_\ell (k_i a_i)},
\end{equation}
where $\gamma_i = \rho_i k /(\rho k_i)$, $a_i$ is the particle radius, $\rho$ is the background density, while $\rho_i$ and $k_i$ are the density and wavenumber of the particle.

\section{Ensemble averaging}
Even if the position and properties of all particles were known, it is still very challenging to solve~\eqref{eq:Individual-expansion-coefficients} for a large number of particles, say, over $10^6$. Also, many sensors can not even measure $f_n^i$, but instead measure the scattered field averaged either in time or space~\cite{foldy_multiple_1945,mishchenko_first-principles_2016}. For these reasons it makes sense to calculate the ensemble average scattered waves. The first step towards achieving this is to introduce a probability for the particles having certain properties and positions~\cite{Linton+Martin2006,Kristensson2015a,Tsang+Kong2001,Tishkovets+etal2011}.

\subsection{Statistical assumptions}
\label{sec:statistics}
To describe the properties and shape of the $i$-th particle, we will use the variable $\lambda_i$, which allows us to define $T_{n}(\lambda_i) := T_{n}^i$ for every $i$.
This means that the $f_n^i$, governed by~\eqref{eq:Individual-expansion-coefficients}, depend on the positions $\rv_1, \, \rv_2, \, \ldots , \, \rv_N$ and the properties $\lambda_1, \, \lambda_2, \, \ldots, \lambda_N$ of all the particles.

To ensemble average we need to assign a probability density for any configuration $\rv_1,\rv_2,\ldots,\rv_N$, and any properties $\lambda_1,\lambda_2,\ldots,\lambda_N$. The first step is consider the $\rv_i$ and $\lambda_i$ as random variables.
Next we assume that the particle properties $\lambda_i$ are sampled from the same domain $\mathcal S$. For example, if $\lambda_i = a_i$, the radius of particle-$i$, for every $i$, then we could choose $\mathcal S = [A_1,A_2]$ so that all  $\lambda_i \in  \mathcal S$, \ie we restrict all particle radii in some interval.
For particle origins $\rv_i$ we can not restrict them all to the same domain because the particles may have a different sizes. So instead we choose a different domain for each, that is, for a given $\lambda_i$ we have that $r_i \in \reg_i$. For example, if all the particles were contained in a sphere with of radius $R$, then a particle with radius $a_i$ would have its origin $\rv_i$ restricted in a sphere of radius $R - a_i$. That is, $\reg_i$ would be a sphere of radius $R - a_i$. For more details on ensemble averaging for multi-species particles see~\cite{gower_reflection_2018}.

 The main parameters we use to describe the average particulate material are
\begin{align} \label{def:number_density}
    &\numdensity(\lambda_i) = \frac{N}{|\reg_i|} p(\lambda_i) \;\; \text{(number of $\lambda_i$ types particles per unit volume)},
    \\
    \label{def:minimal_particle_distance}
    & a_{ij} \;\; \text{(the minimal allowed distance $|\rv_i - \rv_j|$ between particle $i$ and particle $j$)},
\end{align}
where $|\reg_i|$ is the volume of $\reg_i$ and $p(\lambda_i)$ is the probability density of the particle having the property $\lambda_i$. In this paper we allow the minimal distance between two particles $a_{ij}$ to be larger or equal to the sum of the particle radii $a_i + a_j$.
Note we committed an abuse of notation for the function $p$, and will continue to do so.

Let $p(\rv_i, \lambda_i)$ be the probability density of having a particle centred at $\rv_i \in \reg_i$ with $\lambda_i \in \mathcal S$, after ensemble averaging over all other particle positions and properties. If we assume that $\rv_i$ is equally likely to be anywhere in $\reg_i$ we obtain
\begin{equation}
    p(\rv_i, \lambda_i) = p(\rv_i| \lambda_i) p (\lambda_i) \approx
    \frac{p (\lambda_i)}{|\reg_i|} = \frac{\numdensity(\lambda_i)}{N}.
    \label{eqn:p_uniform}
\end{equation}

We also need to define conditional probabilities:
\begin{equation}
    p(\rv_1, \lambda_1; \ldots; \rv_{i-1}, \lambda_{i-1}; \rv_{i+1}, \lambda_{i+1}; \ldots; \rv_{M}, \lambda_M| \rv_i, \lambda_i) = p(\rv_1,\lambda_1;\ldots;\rv_M,\lambda_M)/p(\rv_i, \lambda_i),
    \label{eqn:conditional_prob}
\end{equation}
where $M$ is any integer smaller than the number of particles $N$.

To solve the ensemble average equations, the probability function for two particles $p(\rv_i,\lambda_i;\rv_j,\lambda_j)$ needs to be given. To achieve this, we use an assumption called \emph{hole correction}, which assumes that any two particles are equally likely to be anywhere within regions\footnote{When $\rv_i$ or $\rv_j$ are very close to the boundary of their regions $\reg_i$ and $\reg_j$, then~\eqref{eqn:hole_correction} should be altered. We do not include this alteration because it both does not affect any of the results on effective waves.}, except that their
minimum circumscribed spheres do not overlap~\cite{Fikioris+Waterman1964,Fikioris+Waterman2013}:
\begin{equation} \label{eqn:hole_correction}
    p(\rv_i; \rv_j | \lambda_i;\lambda_j) \approx \begin{cases}
    \frac{1}{|\reg_i| |\reg_j|}
    & \text{for} \;\; |\rv_i - \rv_j | \geq a_{ij},
    \\
    0  & \text{for} \;\; |\rv_i - \rv_j|  < a_{ij}.
    \end{cases}
\end{equation}
To deduce the above for $|\rv_i - \rv_j| \geq a_{ij}$ we used
\begin{equation} \label{eqn:condition-properties}
p(\rv_i; \rv_j | \lambda_i;\lambda_j) = p(\rv_i | \lambda_i;\lambda_j) p(\rv_j |\rv_i, \lambda_i;\lambda_j) \approx \frac{1}{|\reg_i|} \frac{1}{|\reg_j|},
\end{equation}
where $p(\rv_i;\rv_j|\lambda_i;\lambda_j)$ is the probability density of having one particle centred at $\rv_i$, knowing that it has the property $\lambda_i$, and another particle at $\rv_j$, knowing that it has the property $\lambda_j$. The approximation above assumes that the volume of one particle is negligible in comparison to the volume of its confining region.

To help interpret hole-correction~\eqref{eqn:hole_correction} we will do some extra calculations. For simplicity, we assume that the particle properties $\lambda_i$ and $\lambda_j$ are independent of each other to reach
\begin{equation} \label{eqn:hole_correction_conditional}
    p(\rv_j, \lambda_j | \rv_i, \lambda_i) = \frac{p(\rv_i, \lambda_i; \rv_j, \lambda_j)}{p(\rv_i, \lambda_i)} \approx
    |\reg_i|
    p(\lambda_j) p(\rv_i; \rv_j| \lambda_i; \lambda_j)
    \approx \begin{cases}
    \frac{\numdensity(\lambda_j)}{N} & \text{for} \;\; |\rv_i - \rv_j | \geq a_{ij},
    \\
    0  & \text{for} \;\; |\rv_i - \rv_j|  < a_{ij},
    \end{cases}
\end{equation}
where for the last approximation we used~\eqref{eqn:hole_correction} and~\eqref{def:number_density}.
 An alternative way to calculate the above is to approximate $p(\rv_j, \lambda_j| \rv_i,\lambda_i)$ for its expected value in $\rv_i$ and $\lambda_i$ when $|\rv_i - \rv_j| \geq a_{ij}$, that is
\begin{equation} \label{eqn:p_cond_expected}
    p(\rv_j, \lambda_j| \rv_i,\lambda_i) \approx \int_{\reg_i}\int_{\mathcal S} p(\rv_i,\lambda_i) p(\rv_j, \lambda_j| \rv_i,\lambda_i) \mathrm d \rv_i \mathrm d \lambda_i = p(\rv_j, \lambda_j),
    \;\; \text{for} \;\; |\rv_i - \rv_j| \geq a_{ij},
\end{equation}
which when using~\eqref{eqn:p_uniform} leads to the same conclusion as hole correction~\eqref{eqn:hole_correction_conditional}.
Here, $\diff\rv_i$ is the volume measure of the region $\reg_i$.
Later, we show that the quasi-crystalline approximation~\eqref{eqn:quasi_crystalline} makes an approximation which is analogous to~\eqref{eqn:p_cond_expected}.

We can now define the ensemble average of $f_n^1$ as
\begin{equation}
    \ensem{f_n^1} = \int f_n^1 \, p(\rv_1,\lambda_1;\ldots; \rv_N, \lambda_N) \mathrm d \rv_1 \cdots \mathrm d \rv_N \mathrm d \lambda_1 \cdots \mathrm d \lambda_N,
\end{equation}
where the above integrals are over all feasible values for the particles positions $\vec r_i$ and properties $\lambda_i$.

We also need the conditional ensemble averages, which we define as
\begin{align}
    & \ensem{f_{n}^1}(\rv_1,\lambda_1) =  \int f_{n}^1 \, p(\rv_2,\lambda_2;\ldots; \rv_N, \lambda_N |\rv_1, \lambda_1) \mathrm d \rv_2 \cdots \mathrm d \rv_N \mathrm d \lambda_2 \cdots \mathrm d \lambda_N,
    \label{eqn:f1_conditional}
    \\
    \label{eqn:f2_conditional}
    & \ensem{f_{n}^2}(\rv_1,\lambda_1; \rv_2,\lambda_2) =  \int f_{n}^2 \, p(\rv_3,\lambda_3;\ldots; \rv_N, \lambda_N |\rv_1, \lambda_1;\rv_2, \lambda_2) \mathrm d \rv_3 \cdots \mathrm d \rv_N \mathrm d \lambda_3 \cdots \mathrm d \lambda_N.
\end{align}
Note that in~\eqref{eqn:f1_conditional} we are holding the first particle's position $\rv_1$ and properties $\lambda_1$ fixed while averaging
over the other particles. In~\eqref{eqn:f2_conditional} we are averaging $f^2_n$ while holding the first and second particles positions $\rv_1$, $\rv_2$ and properties $\lambda_1$, $\lambda_2$ fixed.

For consistency and simplicity, we will use an approximation for $\ensem{f_n^2}(\rv_1, \lambda_1; \rv_2, \lambda_2)$ which is analogous to both~\eqref{eqn:p_cond_expected} and~\eqref{eqn:hole_correction}, and is called the quasi-crystalline approximation:
\begin{equation}
    \ensem{f_n^2}(\rv_1, \lambda_1; \rv_2, \lambda_2)\approx  \ensem{f_{n}^2}(\rv_2, \lambda_2), \quad \text{for} \;\; |\rv_1 - \rv_2| > a_{12}.
    \label{eqn:quasi_crystalline}
\end{equation}
That is, we replace $\ensem{f_n^2}(\rv_1, \lambda_1; \rv_2, \lambda_2)$ for its expected value in $\rv_1$ and $\lambda_1$, see~\cite{gower_reflection_2018} for a brief discussion on the topic.  This is a standard approach used across statistical physics. It is called a closure approximation~\cite{kuehn_moment_2016,adomian_closure_1971}.

Because the particles only differ due to their position $\rv_i$ and properties $\lambda_i$, we have that $\ensem{f_n^i}(\rv_i,\lambda_i) = \ensem{f_n^j}(\rv_j,\lambda_j)$ for any $i$ and $j$ (all particles with the same properties are indistinguishable).
This is why we now define:
\begin{equation} \label{def:f_n}
    \ensem{f_n}(\rv_j,\lambda_j) := \ensem{f_n^j}(\rv_j,\lambda_j) \quad \text{for } \; j =1,2, \ldots, N.
\end{equation}

\subsection{Average scattered field}

To calculate the ensemble average scattered field we first choose a point $\rv$  outside of the material, where we want to measure the scattered field. For example, turning to Figure \ref{fig:Geometry average}, the point $\rv$ needs to be outside of $\reg_2$ and at least one particle radius $a_2$ away from the boundary of $\reg_2$.
Then we multiple both sides of~\eqref{eqn:total_discrete_wave} by $p(\rv_1,\ldots,\rv_N,\lambda_1,\ldots,\lambda_N)$ and integrate over all possible particle positions and properties to reach
\begin{equation} \label{eqn:total-field}
    \ensem{u(\rv)} = \ui(\rv) + \ensem{\us(\rv)},
\end{equation}
where $\ensem{\ui(\rv)} = \ui(\rv)$, because the incident wave does not depend on the particle configuration, and
\begin{multline}
    \ensem{\us(\rv)} = N \sum_n \int_{\mathcal S} \int_{\reg_1} \ensem{f_n}(\rv_1,\lambda_1) \mathrm u_n (k \rv - k \rv_1) p(\rv_1,\lambda_1) \mathrm d \rv_1 \mathrm d  \lambda_1  \\\approx
     \sum_n  \int_{\mathcal S}  \numdensity(\lambda_1)\int_{\reg_1} \ensem{f_n}(\rv_1,\lambda_1) \mathrm u_n (k \rv - k \rv_1)  \mathrm d \rv_1 \mathrm d  \lambda_1,
    \label{eqn:average-field}
\end{multline}
where we used \eqref{eqn:p_uniform}, \eqref{eqn:conditional_prob}, and \eqref{def:f_n}.
Note that to take the limit $N \to \infty$, it  normally makes sense to fix the number density $ \numdensity(\lambda_i)$ and the probability $p(\lambda_i)$, and then allow the volume of the region $|\reg_i|$ to grow with $N$.

We can rewrite the above when $|\rv| > |\rv_1|$ for every $\rv_1 \in \reg_1$. In this case, we can use the translation matrix~\eqref{eq:translation_spherical_waves} for $\mathrm u_n$ to obtain
\begin{equation}\label{eqn:total-scattered-field}
    \ensem{\us(\rv)}  =  \sum_n \mathfrak F_n \mathrm u_n (k \rv), \;\; \text{with} \;\;  \mathfrak F_n =  \sum_{n'}\int_{\mathcal S}  \numdensity(\lambda_1)\int_{\reg_1}  \mathcal{V}_{n'n}(- k \rv_1)\ensem{f_{n'}}(\rv_1,\lambda_1)\,\diff\rv_1 \diff \lambda_1.
\end{equation}
The $\mathfrak F_n$ are then the average scattering coefficients of the whole material.

\subsection{Average governing equations}

To calculate $\ensem{f_n}(\rv_1,\lambda_1)$,  we need to ensemble average the governing equation~\eqref{eq:Individual-expansion-coefficients}. To achieve this, we set $i=1$, multiple both sides of \eqref{eq:Individual-expansion-coefficients} by  $p(\rv_2,\lambda_2;\rv_3,\lambda_3;\ldots; \rv_N, \lambda_N |\rv_1, \lambda_1)$, and then integrate over all feasible positions and properties while holding $\rv_1$ and $\lambda_1$ fixed.
Then to transform the result into an equation where  $f_n(\rv_1,\lambda_1)$ is the only unknown we use~\eqref{eqn:hole_correction_conditional} and ~\eqref{eqn:quasi_crystalline},  to obtain
\begin{equation} \label{eq:Average-coefficient-system}
    \tcboxmath{
    \begin{aligned}
    &\ensem{f_n}(\rv_1,\lambda_1) = T_{n}(\lambda_1) \sum_{n'} \mathcal{V}_{n'n}(k\rv_1)g_{n'}
    \\&\qquad\qquad+
      T_{n}(\lambda_1) \sum_{n'} \int_{\mathcal S} \bar \numdensity(\lambda_2) \int_{\reg_2 \setminus \mathcal B(\rv_1;a_{12})}\mathcal{U}_{n'n}(k\rv_1 - k\rv_2) \ensem{f_{n'}}(\rv_2,\lambda_2)  \mathrm d \rv_2 \mathrm d \lambda_2,\end{aligned}}
\end{equation}
 for all $\rv_1 \in \reg_1$ and $\lambda_1 \in \mathcal S$,
where we define
\begin{equation}
\mathcal{B}(\rv_1;R) =\{\rv: |\rv -\rv_1| \leq R\},
\label{eqn:B_R}
\end{equation}
used $\reg_2 \setminus \mathcal B(\rv_1;a_{12}) = \{\rv \in \reg_2 : \rv \not\in \mathcal B(\rv_1;a_{12}) \}$ and $\bar \numdensity(\lambda_2) = \frac{N-1}{N} \numdensity(\lambda_2)$.

The system~\eqref{eq:Average-coefficient-system} can be used to solve for $f_n(\rv_1,\lambda_1)$ for any given material geometry $\reg_1$ and any T-matrix $T_n$.
If all particles were the same, \ie same shape and properties, then \eqref{eq:Average-coefficient-system} would be equivalent to~\cite[Equation (4.13)]{Linton+Martin2006} and~\cite[Equation (12)]{Kristensson2015a}. If we considered a two dimensional material, with different types of particles, then~\eqref{eq:Average-coefficient-system} would be equivalent to~\cite[Equation (3.6)]{gower_reflection_2018}.

\begin{optionalnote}{Averaging non-spherical particles}
As a side, we explain how \eqref{eq:Average-coefficient-system} can accommodate particles which are not exactly spherical, as shown in Figure~\ref{fig:CollectionScatterer}.

Assume we have non-spherical particles with a T-matrix $T_{n n'}(\lambda_j, \tau_j)$, which depends on the particle properties $\lambda_j$ and orientation $\tau_j$. If every particle's orientation is statistically independent from everything else\footnote{Including the minimal allowed distance between any two particles~\eqref{def:minimal_particle_distance}.},
then we can set the $T_n$ in~\eqref{eq:Average-coefficient-system} to equal
\[
T_n(\lambda_j) = \int T_{n n}(\lambda_j, \tau_j) p( \tau_j) \mathrm d \tau_j,
\]
where $p( \tau_j)$ is the probability density of particle $j$ being rotated by a $\tau_j$ angle. Note that in general $\tau_j$ could represent three Euler angles. In particular, if the particle is equally likely to be oriented in any direction then its T-matrix $T_{nn'}(\lambda_j,\tau_j)$ averaged over every angle $\tau_j$  becomes diagonal~\cite{mishchenko_t-matrix_1996,varadan_scattering_1979}.

\end{optionalnote}

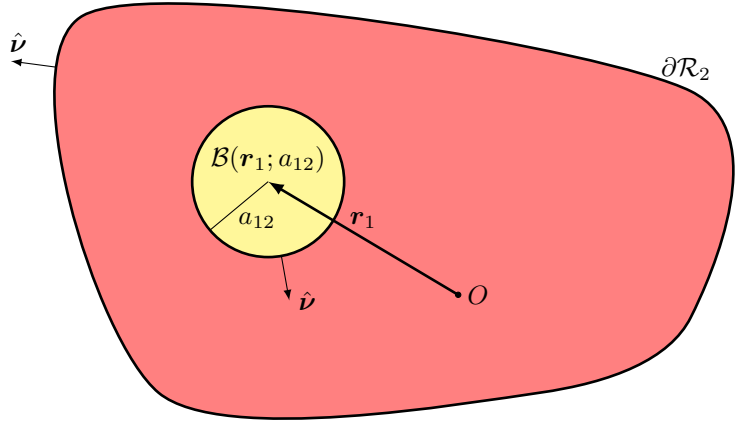
\begin{figure}[t]
    \centering
\begin{tikzpicture}[>=latex, every node/.append style={overlay},scale=1,decoration={markings,mark=at position 0.6 with {\draw[black,->,thin,solid] (0,0) -- (0,-0.6)node[xshift=1mm,above]{$\unitvec{\nu}$};}}]
    \def\domain{(-3,-1) (2,-1) (4,0) (4,3) (-4,4)}
    \def\a{10mm}
    \coordinate (O) at (5,1);
    \coordinate (x) at (2.5,2.5);
    \coordinate (P) at (8,4);

    \filldraw[postaction={decorate},line width=1pt,smooth cycle,fill=red!50,draw=black,shift={(O1)}] plot coordinates \domain;

    \filldraw[line width=1pt,smooth cycle,fill=yellow!50,draw=black,shift={(x)}] circle(\a) node[yshift=3mm] {$\mathcal{B}(\rv_1;a_{12})$};
    \draw[->,thin] (x)++(-80:\a) -- +(-80:6mm) node[right]{$\unitvec{\nu}$};

    \draw[fill=black] (O) circle (1pt) node[right] {$O$};
    \draw[line width=1pt,->] (O)--(x) node[pos=0.5,above] {$\rv_1$};
    \draw (x) -- +(220:\a) node[pos=0.8,right=1mm] {${a_{12}}$};
    \path (P) node {$\partial\reg_2$};
\end{tikzpicture}
    \caption{The geometry of the region $\reg_2$ (red) and the hole correction $\mathcal B(\rv_1;a_{12})$ (in yellow) centred at $\rv_1$. The figure also displays the direction of the surface unit vectors $\nuvh$.}\label{fig:Geometry average}
\end{figure}

\subsection{Symmetry reductions}
\label{sec:symmetry-reductions}
Before solving~\eqref{eq:Average-coefficient-system} to determine the field $\ensem{f_n}(\rv_1,\lambda_1)$, we first look at how to use symmetries to represent $\ensem{f_n}(\rv_1,\lambda_1)$ in a reduced form.

We could apply symmetry reduction directly to the governing integral equation~\eqref{eq:Average-coefficient-system}.
It is, however, simpler to just impose symmetries on the average scattered wave $\ensem{\us(\rv)}$~\eqref{eqn:average-field} and then deduce the resulting symmetry for $\ensem{f_n}(\rv_1,\lambda_1)$
as we demonstrate below. To omit a heavy notation, we will in this section omit the dependence of $\ensem{f_n}$ on $\lambda_1$ and the integrals over the species $\mathcal S$.

{\bf Azimuthal symmetry:} we expect this symmetry when the total scattered wave $\us$ does not change when rotating the measurement point $\rv$ around the $z-$axis. This occurs, for example, for the incident plane wave $\ui(\rv) = \eu^{\iu k z}$ and a spherical material region $\reg_1 = \{|\rv|\leq R: \rv \in \mathbb R^3\}$ centred at the origin. When azimuthal symmetry is present, we expect
\[
\ensem{\us(\mathbf P \rv)} = \ensem{\us(\rv)} \quad \text{for every $|\rv| \not \in \reg_1$},
\]
where we define the operator $\mathbf P$ such that $\mathbf P \rv$ is a $\phi_0$ rotation of the vector $\rv$ around the $z$-axis. Note that when $\reg_1$ is a sphere, then the above should hold true for $\rv \geq R$.

To determine the consequences of this symmetry, we turn to the average scattered wave~\eqref{eqn:average-field} and rewrite in the form
\begin{align*}
    \ensem{\us(\mathbf P \rv)}
    & =  \sum_n \int_{\reg_1} \ensem{f_n}(\rv_1) \mathrm u_n (k \mathbf P \rv - k \rv_1) \mathrm d \rv_1
    =  \sum_n \int_{\reg_1} \ensem{f_n}(\mathbf P  \rv_1) \mathrm u_n (k \mathbf P \rv - k \mathbf P  \rv_1) \mathrm d  \rv_1
    \\
    & =  \sum_n \int_{\reg_1} \ensem{f_n}(\mathbf P \rv_1) \eu^{\iu m \phi_0} \mathrm u_n (k \rv - k  \rv_1) \mathrm d  \rv_1,
\end{align*}
where we used a change of variables $\rv_1\to\mathbf P \rv_1$, and used~\eqref{eqn:outgoing_waves_and_regular_waves} and \eqref{eqn:spherical-harmonics} to substitute $ \mathrm u_n (k  \mathbf P \rv - k  \mathbf P \rv_1) = \eu^{\iu m \phi_0} \mathrm u_n (k \rv - k  \rv_1)$. Notice that the volume measure $\diff\mathbf P\rv_1=\diff\rv_1$.
Equating the above to $\ensem{\us(\rv)}$ and using \eqref{eqn:average-field} then suggests that $\ensem{f_n}(\rv_1) = \ensem{f_n}(\mathbf P \rv_1) \eu^{\iu m \phi_0}$. Then by using a spherical coordinate system $(r_1,\theta_1,\phi_1)$ for $\rv_1$, and by choosing $\phi_0 = - \phi_1$ (without loss of generality) we find that (arguments in spherical coordinates)
\begin{equation} \label{eqn:azimuthal-symmetry}
    \ensem{f_n}(r_1,\theta_1,\phi_1) = \ensem{f_n}(r_1,\theta_1,0) \eu^{-\iu m \phi_1},
\end{equation}
for every $\phi_1$. This symmetry can now be verified by checking that the right hand-side is a solution to~\eqref{eq:Average-coefficient-system}, though this is a longer calculation.

{\bf Planar symmetry:} For an incident plane wave $\ui(\rv) = \eu^{\iu \vec k\cdot\rv}$ and  the material region $\reg_1 = \{\rv \in \mathbb R^3:z>0\}$,  we expect the average scattered wave to satisfy the planar symmetry:
\begin{equation} \label{eqn:planar-symmetry}
    \ensem{\us(\rv)} = \ensem{\us(\rv-\rv_0)} \eu^{\iu \vec k \cdot\rv_0} \quad \text{for every $x,x_0,y, y_0 \in \mathbb R$ and $z <0$},
\end{equation}
where $\rv_0=x_0\xvh+y_0\yvh$. If we then use~\eqref{eqn:average-field} in the above we find that
\begin{align*}
    \ensem{\us(\rv-\rv_0)}  \eu^{\iu k\kvhi\cdot\rv_0} & = \sum_n \int_{\reg_1} \ensem{f_n}(\rv_1) \eu^{\iu k\kvhi\cdot\rv_0} \mathrm u_n (k\rv-k\rv_0-k\rv_1)\,\diff\rv_1
    \\
    & = \sum_n \int_{\reg_1} \ensem{f_n}(\rv_1-\rv_0) \eu^{\iu k\kvhi\cdot\rv_0} \mathrm u_n (k\rv-k\rv_1)  \,\diff\rv_1,
\end{align*}
where for the second equation we changed to the variable of integration $x_1 + x_0\to x_1$ and $y_1 + y_0\to y_1$. For the above to be equal to
\begin{align*}
    \ensem{\us(\rv)} & =  \sum_n \int_{\reg_1} \ensem{f_n}(\rv_1) \mathrm u_n (k\rv-k\rv_1)\,\diff\rv_1,
\end{align*}
for every $x$, $x_0$, $y$, $y_0$, and $z < 0$ suggests that $\ensem{f_n}(\rv_1) = \ensem{f_n}(\rv_1-\rv_0) \eu^{\iu k\kvhi\cdot\rv_0}$, then by choosing $\rv_0 = x_1\xvh+ y_1\yvh$, we find that
\begin{equation} \label{eqn:plane-wave-symmetry}
\ensem{f_n}(\rv_1) = \ensem{f_n}(z_1\zvh) \eu^{\iu k\kvhi\cdot(x_1\xvh+ y_1\yvh)}.
\end{equation}
This symmetry can be verified by checking that the right hand-side is a solution to~\eqref{eq:Average-coefficient-system}.

One case that combines both azimuthal~\eqref{eqn:azimuthal-symmetry} and planar symmetry~\eqref{eqn:plane-wave-symmetry} is the incident plane-wave $\eu^{\iu k z}$ ($\kvhi=\zvh$) and material region $z >0$. In this case (arguments in spherical coordinates),
\begin{equation} \label{eqn:azimuth-planar-symmetry}
    \ensem{f_n}(r_1,\theta_1,\phi_1) = \ensem{f_n}(r_1,0,\phi_0) = \ensem{f_n}(r_1,0,0) \eu^{- \iu m \phi_0}, \quad \text{for every} \; 0\leq \phi_0 \leq 2\pi,
\end{equation}
where the first equation is due to planar symmetry~\eqref{eqn:plane-wave-symmetry} and the second is due to azimuthal symmetry~\eqref{eqn:azimuthal-symmetry}. Equation~\eqref{eqn:azimuth-planar-symmetry} can only be true for every $\phi_0$ when
\[
\ensem{f_n}(\rv_1) = \delta_{m,0} \ensem{f_\ell}(z_1\zvh).
\]
This result will be used later to reach a simplified dispersion equation.


\section{Effective wavenumbers}

\subsection{Wave decomposition}
\label{sec:wave-decomposition}
Much of the literature has focused on solving~\eqref{eq:Average-coefficient-system} by assuming that the unknown field $f_n(\rv_1, \lambda_1)$ satisfies a wave equation for the spatial variable $\rv_1$ and for some effective wavenumber $k_*$. This assumption implies that the average transmitted fields $\ensem{u(\rv)}$ also satisfy a wave equation \cite{martin_multiple_2011}.
Recent results \cite{gower2019multiple,gower2019proof} have demonstrated that  for a half-space, the exact solution for $f_n(\rv_1, \lambda_1)$ is a sum of plane waves, each with a different wavenumber. Here we generalise this result by considering $f_n(\rv_1, \lambda_1)$ to be a sum of isotropic waves of any type, \ie not necessarily a plane wave, and the particulate material to occupy any region.

In general, we propose the representation
\begin{align}
   & \ensem{f_n}(\rv_1, \lambda_1) = \sum_p f_{p,n}(\rv_1, \lambda_1) \quad \text{with} \quad  \nabla^2_{\rv_1} f_{p,n}(\rv_1,  \lambda_1)  = -k_p^2 f_{p,n}(\rv_1, \lambda_1), \quad \text{for} \;\; \rv_1 \in \reg_1,
   \label{eqn:effective_wave_representation}
\end{align}
where the Laplacian $\nabla^2_{\rv_j}$ is taken in terms of $\rv_j$.

Our first major result is to calculate the effective wavenumbers $k_1$, $k_2$, $\cdots$, and to demonstrate that they depend only on the particle properties, and not on the geometry of the region enclosing the particles. Although the geometry of material and the incident wave will determine which of these wavenumbers are excited. A another major result, is that most of the effective wavenumbers are highly attenuating, which implies that the series~\eqref{eqn:effective_wave_representation} rapidly converges to the exact solution.
In the remained of this section, we show how to deduce a system that determines the $f_{p,n}(\rv_1,\lambda_1)$ that is decoupled from the material geometry.

To simplify the governing system~\eqref{eq:Average-coefficient-system} we note that by definition~\eqref{eq:translation_spherical_waves} the translation matrix $\mathcal{U}_{n'n}(k\rv_1 - k\rv_2)$ satisfies a wave equation in either $\rv_1$ or $\rv_2$ with wavenumber $k$. This and the representation~\eqref{eqn:effective_wave_representation} leads to
\begin{multline}
    (k^2 - k_p^2) \mathcal{U}_{n'n}(k\rv_1 - k\rv_2) f_{p,n'}(\rv_2,\lambda_2)
    \\=
    \mathcal{U}_{n'n}(k\rv_1 - k\rv_2) \nabla_{\rv_2}^2 f_{p,n'}(\rv_2,\lambda_2) - \nabla_{\rv_2}^2 \mathcal{U}_{n'n}(k\rv_1 - k\rv_2) f_{p,n'}(\rv_2,\lambda_2),  \quad \text{for} \;\; \rv_1 \in \reg_1,
     \label{eqn:wave_difference}
\end{multline}
Then for $\rv_1 \in \reg_1(a_{12})= \{\rv \in\reg_1 : d(\rv,\partial \reg_1) \geq {a_{12}} \}$, we can integrate both sides over $\rv_2 \in\reg_2 \setminus \mathcal B(\rv_1;a_{12})$ and apply Green's second identity to reach,
\begin{equation}
    \int_{\reg_2\setminus \mathcal B(\rv_1;a_{12})} \mathcal{U}_{n'n}(k\rv_1 - k\rv_2) f_{p,n'}(\rv_2,\lambda_2) \mathrm d \rv_2 =
     \frac{\mathcal I_p (\rv_1) - \mathcal J_p(\rv_1)}{k^2 - k_p^2},\quad \rv_1 \in \reg_2(a_{12})
     \label{eqn:greens_2nd}
\end{equation}
where
\begin{align}
     & \mathcal I_p (\rv_1) = \int_{\partial \reg_2 }  \mathcal{U}_{n'n}(k\rv_1 - k\rv_2) \frac{\partial f_{p,n'}(\rv_2,\lambda_2)}{\partial \vec\nu_2}
     - \frac{\partial \mathcal{U}_{n'n}(k\rv_1 - k\rv_2) }{\partial \vec\nu_2} f_{p,n'}(\rv_2,\lambda_2) \mathrm dA_2,
     \label{eqn:Ireg}
     \\
     & \mathcal J_p(\rv_1) =
     \int_{\partial B(\vec 0;a_{12})}  \mathcal{U}_{n'n}(-k\rv) \frac{\partial f_{p,n'}(\rv + \rv_1,\lambda_2)}{\partial \vec\nu}
     - \frac{\partial \mathcal{U}_{n'n}(-k \rv) }{\partial \vec\nu} f_{p,n'}(\rv + \rv_1,\lambda_2)\mathrm dA,
     \label{eqn:Ib}
\end{align}
$\mathrm d A_2$ and $\mathrm dA$ are the surface elements for $\rv_2$ and $\rv$, respectively, and $\boldsymbol\nu_2$ and $\boldsymbol\nu$ are outward pointing normal vectors to the surfaces $\partial \reg_2$ and $\partial  B(\vec 0;a_{12})$, respectively {(see Figure~\ref{fig:Geometry average})}.
To reach~\eqref{eqn:Ib} we changed the integration variable to $ \rv = \rv_2 - \rv_1$.

Both integrals $\mathcal I_p(\rv_1)$ and $\mathcal J_p (\rv_1)$ depend on the indices $n,n'$ and on the state variable $\lambda_1$ and $\lambda_2$, which we omit to avoid a heavy notation.

\begin{optionalnote}{The boundary layer}
The region $\reg_1 \setminus \reg_1(a_{12})$ is often called a boundary layer.  The simplification~\eqref{eqn:greens_2nd} only occurs when $\rv_1$ is not in this boundary layer, that is when $\rv_1 \in \reg_2(a_{12})$. This is because then $\partial \reg_2 \cap \partial \mathcal  B(\rv;a_{12}) = \varnothing$, and therefore the boundary of $\reg_2 \setminus \mathcal B(\rv_1;a_{12})$ becomes $\partial \reg_2 \cup \partial \mathcal B(\rv_1;a_{12})$.

In this paper we will not discuss how to evaluate the system~\eqref{eq:Average-coefficient-system} for $\rv_1 \in \reg_1 \setminus \reg_1(a_{12})$. Evaluating in this boundary layer is used to combine the different fields $f_{p,n}$, but this is only needed for very strong multiple-scattering~\cite{gower2019proof,gower2019multiple}.
{We also note that for a slab geometry, it is possible to evaluate the complex integrals when $\rv_1$ is located in the boundary layer, see~\cite{Kristensson2015a,Kristensson2015b}.}
\end{optionalnote}

By substituting equations~\eqref{eqn:effective_wave_representation}--\eqref{eqn:Ib} into the governing equation~\eqref{eq:Average-coefficient-system}, and assuming $\rv_1 \in \reg_2(a_{12})$, we obtain
\begin{equation}
   \sum_p f_{p,n}(\rv_1,\lambda_1) = \sum_{n'} T_{n}(\lambda_1)\mathcal{V}_{n'n}(k\rv_1)g_{n'}
    + \sum_{n' p}T_{n}(\lambda_1) \int_{\mathcal S}
     \frac{\mathcal I_p (\rv_1) - \mathcal J_p(\rv_1)}{k^2 - k^2_p}
     \bar\numdensity(\lambda_2) \mathrm d \lambda_2, \quad \rv_1 \in \reg_2(a_{12}),
    \label{eqn:combined_green_waves}
\end{equation}
The key to simplifying~\eqref{eqn:combined_green_waves} is to note that both
$f_{p,n}(\rv_1)$ and $\mathcal J_p(\rv_1)$ satisfy a wave equation with wavenumber $k_p$ and spatial position $\rv_1$, whereas $\mathcal{V}_{n'n}(k\rv_1)$ and $\mathcal I_p (\rv_1)$ satisfy a wave equation with wavenumber $k$. This enables us to use Theorem~\ref{th:Vandermonde} to conclude that
\begin{align} \label{eqn:ensemble_wave}
   &
   \tcboxmath{
   f_{p,n}(\rv_1,\lambda_1) +  \sum_{n'} \frac{T_{n}(\lambda_1)}{k^2 - k^2_p} \int_{\mathcal S}  \mathcal J_p (\rv_1)  \bar \numdensity(\lambda_2) \mathrm d \lambda_2 = 0,
   }
   \quad \text{(the ensemble wave equation)},
\\
  & \tcboxmath{
  \sum_{n'}
  \mathcal{V}_{n'n}(k\rv_1)g_{n'}  +
    \sum_{n'p}\int_{\mathcal S} \frac{\mathcal I_p (\rv_1)}{k^2 - k^2_p}
 \bar  \numdensity(\lambda_2) \mathrm d \lambda_2
     = 0,
  }  \quad \text{(the ensemble boundary conditions)}
      \label{eqn:general_boundary}
\end{align}
both valid for $\rv_1 \in \reg_1( a_{12})$.
Now it is clear that~\eqref{eqn:ensemble_wave} is independent of both \emph{the region of particles} $\reg_j$ and \emph{the incident field}.  We will show how the effective wavenumbers $k_p$ can be completely determined from~\eqref{eqn:ensemble_wave}. On the other hand, equation~\eqref{eqn:general_boundary} depends on both the region and incident wave, and will lead to a restriction on how to combine the $f_{p,n}(\rv,\lambda)$. Equation~\eqref{eqn:general_boundary} is sometimes called the extinction equation.

Both~\eqref{eqn:ensemble_wave} and~\eqref{eqn:general_boundary} can be further simplified by expanding the fields in terms of orthonormal functions,
which we do in the section below.


\section{The dispersion equation}\label{sec:dispersion_equations}

The effective wavenumbers $k_p$ and much about the fields $f_{p,n}(\rv_1, \lambda_1)$, can be calculated just from the ensemble wave equation~\eqref{eqn:ensemble_wave}. Depending on the symmetries of $f_{p,n}(\rv_1, \lambda_1)$ we can reach different dispersion equations, the most general of which just assumes that $f_{p,n}(\rv_1, \lambda_1)$ is a smooth field.

\subsection{Effective regular waves}
\label{sec:regular_dispersion}
Here, we determine the effective wavenumbers $k_p$ from the ensemble wave equation~\eqref{eqn:ensemble_wave}. To do so, we use an origin $O$, for our coordinate system, located in $\reg_2(a_{12})$,
and we expand $f_{p,n}(\rv_1,\lambda_1)$ in a series of regular {spherical} functions of the wave equation in the domain $\reg_2(a_{12})$.

This regular series takes the form:
\begin{align}
\label{eqns:fields_spherical_basis}
   & f_{p,n}(\rv_1,\lambda_1) = \sum_{n_1} F_{p,nn_1}(\lambda_1) \mathrm v_{n_1} (k_p\rv_1),
    \\
    & f_{p,n}(\rv+\rv_1,\lambda_2) = \sum_{n_1n_2} F_{p,nn_1}(\lambda_2) {\mathcal V}_{n_1n_2}(k_p\rv) \mathrm v_{n_2} (k_p\rv_1),
\end{align}
where the coefficients $F_{p,nn_1}(\lambda_1)$ are to be determined, and we used the translation matrix ${\mathcal V}_{n_1n_2}(k_p\rv)$ in the second expansion.
{The convergence of this series depends on the behaviour of the expansion coefficients $F_{p,nn_1}(\lambda_1)$, which in turn depend on both the confining geometry and the incident field. At this stage, we assume the series is convergent.}
In Appendix~\ref{sec:Wronskian} we show how substituting the above into the dispersion equation~\eqref{eqn:ensemble_wave} leads to
\begin{equation} \label{eqn:regular_basis}
\sum_{n_1} F_{p,nn_1}(\lambda_1) \mathrm v_{n_1} (k_p\rv_1) + \sum_{n_2} G_{n,n_2}(\lambda_1)\mathrm v_{n_2} (k_p\rv_1)=0,
\end{equation}
where
\begin{equation} \label{eqn:G-simple}
     G_{n,n_2}(\lambda_1) =
      \frac {T_{n}(\lambda_1)}{k^2_p - k^2}  \sum_{n_1 n_3 n'}c_{n n'n_3} c_{n_1 n_2n_3}  \int_{\mathcal S}  a_{12} \mathrm{N}_{\ell_3}(ka_{12},k_p a_{12}) F_{p,n'n_1}(\lambda_2)
     \numdensity(\lambda_2) \diff \lambda_2
\end{equation}
where the $c_{n' n n_1}$ are numbers which are defined in Appendix~\ref{sec:Translation}, and
\begin{equation} \label{eqn:N_ell}
    \mathrm {N}_{\ell}(x,z) = x{\mathrm{h}_\ell^{(1)\prime}}(x)\mathrm{j}_\ell(z) - z\mathrm{h}_\ell^{(1)}(x)\mathrm{j}_\ell'(z).
\end{equation}
By using the orthonormal property of spherical harmonics \eqref{eqn:regular_basis} reduces to
\begin{equation} \label{eqn:regular-eigensystem}
    \tcboxmath{
    F_{p,nn_2}(\lambda_1) +  G_{n, n_2}(\lambda_1) = 0.
    }
    \quad \text{(the regular eigen-system)}
\end{equation}
As the above is a linear system of equations for the unknowns $F_{p,n'n_1}(\lambda_1)$, we can rewrite it in the form\footnote{Note we would have to discretise the integral over $\mathcal S$ to reach this determinant equation.}
\begin{equation}
(\mathbf I+\mathbf G) \mathbf F=0 \implies \det (\mathbf I + \mathbf G) = 0,
\label{eqn:regular-dispersion}
\end{equation}
where the second equations holds for a non-zero $\mathbf F$. This determinant equation can be used to find all effective wavenumbers $k_p$ for any geometry. Although~\eqref{eqn:regular-dispersion} contains all possible effective wavenumbers it is computational simpler to solve the planar dispersion equation which also contains all viable effective wavenumbers, as we will show below. Solving~\eqref{eqn:regular-dispersion} can be numerically difficult for two reasons: 1) the roots of \eqref{eqn:regular-dispersion} are multiple roots with different multiplicities, and 2) there are many spurious roots, as discussed in the optional box below.

The plane wave dispersion, and other reduced dispersion equations, are calculated by restricting the form of $f_{p,n}(\rv_1,\lambda_1)$ through the use of symmetry reductions as shown in Section~\ref{sec:symmetry-reductions}.

\begin{optionalnote}{Spurious wavenumbers $k_p$}
One problem with using~\eqref{eqn:regular-dispersion} to find the wavenumbers $k_p$ is that~\eqref{eqn:regular-dispersion} has spurious roots. That is, it has solutions $k_p$ which are not solutions to the ensemble wave equation~\eqref{eqn:ensemble_wave}. Figure ? shows some of these spurious solutions.

These spurious solutions appear when truncating the index $n_2$ in $F_{p,nn_2}$ and then solving \eqref{eqn:regular-dispersion}. Calculating the eigenvectors $F_{p,nn_2}$ of these spurious solutions $k_p$ would then lead the series in~\eqref{eqns:fields_spherical_basis} to rapidly diverge, which is physically not viable. That is, when truncating the system~\eqref{eqn:regular-eigensystem} for $\ell_2 \leq L_2$, it would also be natural to truncate $\ell_1 \leq L_2$. However, neglecting the terms $F_{p,n' n_1}$ with $\ell_1 > L_2$ on the right hand-side of~\eqref{eqn:regular-eigensystem} is only approximately correct if $|F_{p,n' n_1}|$ is small or at least getting smaller when increasing $\ell_1$. These spurious roots $k_p$ on the other hand lead to $|F_{p,n' n_1}|$ which increase with $\ell_1$.
\end{optionalnote}

\subsection{Effective azimuthal waves}\label{sec:azimuthal-eigensystem}
When both the incident wave and material region share a rotational symmetry around the $z-$axis we can reach a reduced representation of $f_{p,n}(\rv_1,\lambda_1)$ by using~\eqref{eqn:azimuthal-symmetry}. For example, this occurs when the incident wave is $\ui(\rv) = \eu^{\iu k z}$ and the material region is a sphere centred at the origin.

Combining the symmetry~\eqref{eqn:azimuthal-symmetry} with the representation~\eqref{eqns:fields_spherical_basis} leads to the form
\begin{equation} \label{eqn:f-azimuthal}
    f_{p,n}(\rv_1,\lambda_1) = \sum_{\ell_1 \geq |m|} F_{p,n \ell_1}(\lambda_1) \mathrm v_{(\ell_1,-m)} (k_p\rv_1),
\end{equation}
{where $F_{p,n \ell_1} = F_{p,n(\ell_1,-m)}$. In} other words, in terms of the representation \eqref{eqns:fields_spherical_basis} we have that $F_{p,nn_1} = 0$ unless $m_1 = -m$ and $\ell_1 \geq |m|$.

Substituting~\eqref{eqn:f-azimuthal} into~\eqref{eqn:regular-eigensystem} then leads to
\begin{equation} \label{eqn:dispersion-azimuthal}
F_{p,n \ell_2}(\lambda_1)
+ G_{ n, (\ell_2, -m)}(\lambda_1) = 0 \quad \text{for every } n, \, \ell_2 \geq |m|,
\end{equation}
which is a restricted version of~\eqref{eqn:regular-dispersion}. Note that for the above, in the sum~\eqref{eqn:G-simple} we should set $m_1 = -m'$ and sum over $\ell_1 \geq |m'|$. Further~\eqref{eqn:regular-eigensystem} also leads to
$G_{ n, n_2}(\lambda_1) = 0$
for every $m_2 \not = -m$, but this is automatically satisfied because  $c_{nn'n_3}c_{(\ell_1,-m') (\ell_2,m_2) n_3} = 0$ when $m_2 \not = -m$.

\subsection{Effective plane-waves}\label{sec:plane_waves}

Here, we restrict $f_{p,n}(\rv_1,\lambda_1)$ by imposing planar symmetry~\eqref{eqn:plane-wave-symmetry}. This will allow us to deduce simpler dispersion equations, as well as deduce reflection and transmission from a plate.

By combining the symmetry~\eqref{eqn:plane-wave-symmetry} with the wave equation~\eqref{eqn:effective_wave_representation} we first conclude that
\begin{equation} \label{eqn:effective-plane-wave}
    f_{p,n}(\rv_1,\lambda_1) = F_{p,n}(\lambda_1) \eu^{\iu \vec k_p \cdot \rv_1},
\end{equation}
where $\vec k_p \cdot \rv_1$ is simply a sum of element wise multiplication without conjugation, even though $\vec k_p$ is a complex vector.

In the appendix~\ref{sec:Effective plane waves} we show that substituting the above into~\eqref{eqn:ensemble_wave} leads to
\begin{equation}  \label{eqn:planewave_eigensystem}
\tcboxmath{
    F_{p,n}(\lambda_1) +  \sum_{n'n_1}  \frac{4 \pi c_{n' nn_1}}{k^2_p - k^2}  \iu^{-\ell_1}  \mathrm Y_{n_1}(\unitvec{k}_p)T_{n}(\lambda_1) \int_{\mathcal S} a_{12} \mathrm{N}_{\ell_1} (k a_{12}, k_p a_{12})F_{p,n'}(\lambda_2)  \bar \numdensity(\lambda_2) \mathrm d \lambda_2 = 0,
}
\end{equation}
and that, when considering only one type of particle, the above  reduces to an equation which is found in much of the literature~\cite{linton_multiple_2006,Fikioris+Waterman1964,Tsang+Kong2001,Mackowski2001,Doicu+Mishchenko2019a,gower_reflection_2018}.


Note that
\begin{equation}\label{eqn:Complex angles}
    \unitvec{k}_p = (\sin \theta_p \cos \phi_p,\sin \theta_p\sin \phi_p,\cos \theta_p) \quad \text{and} \quad Y_{n_1}(\unitvec{k}_p) = Y_{n_1}(\theta_p,\phi_p),
\end{equation}
and the angles $\theta_p$ and $\phi_p$ can be complex numbers, meaning that we may have $|\unitvec{k}_p| \not = 1${, but we do have that $\unitvec{k}_p\cdot\unitvec{k}_p=1$} for the real inner product.

Equation~\eqref{eqn:planewave_eigensystem} can be turned into a determinant equation, much like~\eqref{eqn:regular-dispersion}, from which we can calculate effective wavenumbers $k_p$:
\begin{equation}
(\mathbf I+\mathbf C) \mathbf F=0 \implies \det (\mathbf I + \mathbf C) = 0.
\label{eqn:plane-wave-dispersion}
\end{equation}
The form of the eigensystem~\eqref{eqn:planewave_eigensystem} seems to suggest that the $k_p$ depend on $\unitvec{k}_p$. However, a direct (though cumbersome) evaluation of the resulting (truncated) determinant system would confirm that $\unitvec{k}_p$ has no contribution. Further, the more general eigensystem~\eqref{eqn:regular-eigensystem} does not depend on $\unitvec{k}_p$. As a sanity check, we can explicitly show that every solution to~\eqref{eqn:planewave_eigensystem} is also a solution to \eqref{eqn:regular-eigensystem}. To achieve this we rewrite the solution~\eqref{eqn:effective-plane-wave} in the form~\eqref{eqns:fields_spherical_basis} by using a plane-wave expansion~\eqref{eqn:plane-wave-expansion}
in~\eqref{eqn:effective-plane-wave} to obtain:
\begin{equation} \label{eqn:planewave-expansion}
   f_{p,n}(\rv_1,\lambda_1) = \sum_{n_1} F_{p,n n_1}(\lambda_1)  \mathrm v_{n_1}(k_p \rv_1) \quad \text{with} \quad
   F_{p,n n_1}(\lambda_1) = 4 \pi \iu^{\ell_1} (-1)^{m_1} F_{p,n}(\lambda_1) \mathrm Y_{(\ell_1,-m_1)}(\unitvec{k}_p).
   \end{equation}
Because the above is in the form~\eqref{eqns:fields_spherical_basis} and the field $f_{p,n}(\rv_1,\lambda_1)$ satisfies the general dispersion equation~\eqref{eqn:ensemble_wave} (when~\eqref{eqn:planewave_eigensystem} is satisfied) then $F_{p,n n_1}$ and $k_p$ must also satisfy the regular eigensystem~\eqref{eqn:regular-eigensystem}.

As the $k_p$ are independent of $\unitvec{k}_p$, we can choose any $\unitvec{k}_p$ to calculate the $k_p$. We exemplify for a single species: take $\unitvec{k}_p = \unitvec{z}$ so that $Y_{n_1}(\unitvec{k}_p) = \sqrt{2 \ell_1 +1}/\sqrt{4\pi} \delta_{m_1 0}$, then the $k_p$ must satisfy
\begin{align}  \label{eqn:planewave_det}
    & \det ( M_{n n'}(k_p) ) = 0, \quad \text{where}
    \\ \notag
    & M_{n n'}(k_p) = \delta_{n n'} +  \sum_{\ell_1}  \frac{\sqrt{4\pi} c_{n' n (\ell_1,0)}}{k^2_p - k^2} \iu^{-\ell_1} \sqrt{2 \ell_1 +1}  T_{n}  \bar \numdensity a_{12} \mathrm{N}_{\ell_1} (k a_{12}, k_p  a_{12}  ) = 0,
\end{align}
for a single species\footnote{We chose not to show the multi-species version because it would require discretising the integral over $\lambda_2$.}, where $\bar \numdensity = \numdensity (N-1)/N$ and $\numdensity$ is the number density of particles.
This equation can be even further simplified  when there is azimuthal symmetry, as is the case for the incident plane-wave $\eu^{\iu k z}$ and material region $z >0$. In this case we can apply the symmetry~\eqref{eqn:azimuth-planar-symmetry} to ~\eqref{eqn:planewave_det} and reach
\begin{equation} \label{eqn:azi-planewave_det}
\det ( M_{(\ell,0) (\ell',0)}(k_p) ) = 0.
\end{equation}
All the $k_p$ that satisfy~\eqref{eqn:azi-planewave_det} also satisfy~\eqref{eqn:planewave_det}, however, there are solutions to~\eqref{eqn:planewave_det} which do \emph{not} satisfy~\eqref{eqn:azi-planewave_det}, see Figure~\ref{fig:wavenumbers-compare}. That is, it is not possible to excite all effective wavenumbers when considering only direct incidence $\ui(\rv) = \eu^{\iu k z}$. In other words, one type of experiment (one type of incident wave and material geometry) can only excite a portion of all the effective wavenumbers\footnote{Note this does \emph{not} mean there exist wavenumbers $k_p$ that change (continuously) with the angle of incidence $\unitvec{k}_p$. There are, however, solutions ${k}_p$ which are not excited for certain angles of incidence.}.

\subsection{Plane-wave dispersion has all viable effective wavenumbers}
\label{sec:plane-wave-disp-all}

Because the representation~\eqref{eqns:fields_spherical_basis} is more general than a plane-wave representation~\eqref{eqn:effective-plane-wave}, we know that all solutions $k_p$ to the plane-wave dispersion~\eqref{eqn:plane-wave-dispersion} must also satisfy the more general regular dispersion~\eqref{eqn:regular-eigensystem}. There is even an explicit conversion from plane-wave solutions to the regular solutions~\eqref{eqn:planewave-expansion}. However, it is not at all obvious that all viable solutions $k_p$ of the regular dispersion~\eqref{eqn:regular-eigensystem} must satisfy plane-wave dispersion~\eqref{eqn:plane-wave-dispersion}.

To show that all viable wavenumbers $k_p$ that satisfy the regular dispersion~\eqref{eqn:regular-dispersion} must also satisfy the plane-wave dispersion~\eqref{eqn:plane-wave-dispersion}, we need to rewrite the expansion~\eqref{eqns:fields_spherical_basis} in terms of plane-waves. 
We achieve this by using
\begin{equation}\label{eqn:Reg wave in plane waves}
    \mathrm v_{n_1} (k_p\rv_1)=\frac{1}{4\pi\iu^{\ell_1}}\int_{\Omega_{q}}  \mathrm Y_{n_1}(\unitvec{q})\eu^{\iu  k_p \unitvec{q} \cdot \rv}\,\diff \Omega_{q},
\end{equation}
where $\Omega_q$ is the solid angle of the radial unit vector $\unitvec{q}$.
The above can be verified by using a plane-wave expansion~\eqref{eqn:plane-wave-expansion} for $\eu^{\iu k_p \unitvec{q} \cdot \rv_1}$, to write \begin{equation}
    f_{p,n}(\rv_1,\lambda_1) = \int_{\Omega_{q}} F_{p,n}(\unitvec{q}, \lambda_1)  \eu^{\iu  k_p \unitvec{q} \cdot \rv_1}\,\diff \Omega_{q},
\end{equation}
where
\begin{equation}
     F_{p,n}(\unitvec{q}, \lambda_1) = \frac{1}{4\pi} \sum_{n_1} F_{p,nn_1}(\lambda_1) \iu^{-\ell_1}\mathrm Y_{n_1}(\unitvec{q}).
\end{equation}
Written in this form, $f_{p,n}(\rv_1,\lambda_1)$ is now a superposition of plane waves all with the same wavenumber $k_p$ but with different directions $\unitvec{q}$. Note that the $F_{p,nn_1}$ depend on $k_p$ but are independent of the $\unitvec{q}$. To find a dispersion equation we can repeat the same steps that led to~\eqref{eqn:planewave_eigensystem} to reach:
\begin{multline}  \label{eqn:planewave_eigensystem_integrated}
    \int_{\Omega_{q}} \left [F_{p,n}(\unitvec{q}, \lambda_1) +  \sum_{n'n_1}  \frac{4 \pi c_{n' nn_1}}{k^2_p - k^2} \iu^{-\ell_1}  \mathrm Y_{n_1}(\unitvec{q})T_{n}(\lambda_1) \int_{\mathcal S} a_{12} \mathrm{N}_{\ell_1} (k a_{12}, k_p a_{12})F_{p,n'}(\unitvec{q},\lambda_2)  \bar \numdensity(\lambda_2) \diff \lambda_2 \right ]
    \\
    \times \eu^{\iu k_p \unitvec{q} \cdot \rv_1} \diff \Omega_{q}  = 0.
\end{multline}
As the map $\unitvec{q} \to F_{p,n}(\unitvec{q},\lambda_1)$ is smooth, we can conclude that the integrand in the above is also a smooth function of $\unitvec{q}$, in which case, the above can only be zero for every $\rv_1$ when the integrand is zero for every $\unitvec{q}$. That is, the $k_p$ and $F_{p,n}(\unitvec{q},\lambda_1)$ have to satisfy the plane-wave dispersion equation~\eqref{eqn:planewave_eigensystem}, with the same $k_p$ for every $\unitvec{q}$, which in turn implies that $k_p$ has to satisfy the determinant equation~\eqref{eqn:plane-wave-dispersion}.\footnote{In more detail, let
$\int_{\Omega_{k}} f(\unitvec{k}) \eu^{\iu k \unitvec{k} \cdot \rv} \diff \Omega_{k}  = 0,\forall \rv\in\reg$, and expand $f(\unitvec{k})$ in spherical harmonics, \ie $f(\unitvec{k})=\sum_n f_n{\mathrm Y}_n(\unitvec{k})$.
We assume this relation holds in a ball of radius $R$, centred at the origin.
Then, using the transformation~\eqref{eqn:Reg wave in plane waves}, we obtain
$\sum_n \iu^\ell f_n{\mathrm v}_n(k\rv)=0,\forall\rv\in\reg$.
Orthogonality of the spherical harmonics implies $f_n{\mathrm j}_\ell(kr)=0,r\in[0,R],\forall n$, and $f_n=0$, since the zeros of the spherical Bessel functions are isolated points on the real axis.
}


\subsection{Effective properties in the long wavelength limit} \label{sec:effective-properties}
By taking the limit where the incident wavelength is long compared to the particle diameter, we can calculate the effective properties directly from any of the dispersion equations. This procedure is explained in detail in~\cite{parnell_multiple_2010,martin_estimating_2010,gower_reflection_2018}.

It is particularly interesting to calculate the effective properties from the regular dispersion equation~\eqref{eqn:regular-dispersion}, because this equation holds for any material geometry, which then gives us confidence that the effective properties are truly properties of the material's microstructure and medium, and not its geometry.

Here we calculate the effective properties for spherical particles~\eqref{eqn:T-matrix-acoustics} and acoustics. To achieve this, we need to consider the limit $k \to 0$, starting with the T-matrix coefficients which scale with $k$ in the form
\begin{equation} \label{eqns:T-matrix-low-freq}
    T_{(0,0)}(\lambda_j) \sim \frac{\iu k^3a_j^3 }{3} \Delta \beta_j, \quad
    T_{(1,m)}(\lambda_j) \sim -\frac{\iu k^3a_j^3 }{3} \Delta \rho_j \quad \text{for}\;\; m=-1,0,1,
\end{equation}
and $T_{(\ell,m)}(\lambda_j) \sim 0$ for $\ell > 1$, where
\[
\Delta \rho_j = \frac{\rho - \rho_j}{\rho + 2 \rho_j} \quad \text{and} \quad
\Delta \beta_j = \frac{\beta - \beta_j}{\beta_j},
\]
with $\beta_j = \rho_j \omega^2 / k_j^2$ and $\beta = \rho \omega^2 / k^2$  being the bulk modulus of the $j$-th particle and of the background medium, respectively.
For particles with any shape, there is a similar result for their T-matrix when assuming that they scatter only monopole and dipole waves\cite{Varadan+Varadan1986}.

To facilitate the next steps, we rewrite the effective wavenumber $k_*$ and expand
\[
k_* = \frac{k}{c_*} \sqrt{\frac{\beta}{\rho}}
\quad \text{and} \quad
\mathrm N_{\ell}(k a_{12},k_* a_{12}) \sim \frac{\iu}{k a_{12} c_*^\ell} \left(\frac{\beta}{\rho} \right)^{\ell/2},
\]
where $c_*$ is the constant effective phase speed.

The expansion~\eqref{eqns:T-matrix-low-freq} for the $T_n$ imply that, at leading order in small $k$, only $F_{p,n n_1}$ for $n = (0,0),(1,-1),(1,0),(1,1)$ has a significant contribution, with all other terms being zero. We will also truncate $\ell_1$ in $F_{p,n n_1}$ by assuming $\ell_1 \leq L$ for some $L \geq 2$.

By substituting the above into~\eqref{eqn:regular-dispersion}, and then expanding up to leading order in small $k$, we find \emph{three} possible solutions for the effective phase speed $c_*$. Two are these solutions are non-physical, because they do not satisfy the plane-wave dispersion, as discussed in Section~\ref{sec:plane-wave-disp-all} and at the end of Section~\ref{sec:regular_dispersion}. The only remaining physically viable solution is
\begin{equation} \label{eqn:effective-phasespeed}
\tcboxmath{
c_*^2 = \frac{\beta}{1 + \ensem{ \Delta \beta_j } } \frac{1}{\rho}\frac{1 + 2 \ensem{\Delta \rho_j}  }{1 - \ensem{\Delta \rho_j} },
} \quad \text{(effective phase speed)}
\end{equation}
where, just for this section, we define
\[
\ensem{\Delta \beta_j}  =  \int_{\mathcal S} \Delta\beta_j  \varphi(\lambda_j) d \lambda_j,
\]
and likewise for $\Delta \rho_j$, where $\varphi(\lambda_j)\diff \lambda_j$ is the volume fraction of particles with the properties $\lambda_j$, so that $\int_{\mathcal S}\varphi(\lambda_j)\diff \lambda_j = \varphi$, the total particle volume fraction.
Note that for spheres $
\varphi(\lambda_j) = \frac{4 \pi a_j^3}{3} \numdensity(\lambda_j)$.


By writing $c_*^2 = \beta_*/\rho_*$ we can now identify the effective bulk modulus $\beta_*$ and density $\rho_*$ as
\begin{equation} \label{eqns:effective-properties}
    \tcboxmath{
    \frac{1}{\beta_*} = \frac{1}{\beta} + \frac{\ensem{\Delta \beta_j}}{\beta}
    }
    \quad \text{and} \quad \tcboxmath{\rho_* = \rho \frac{1 -\ensem{\Delta \rho_j}}{1 + 2 \ensem{\Delta \rho}},
    }
    \quad \text{(effective properties)}
\end{equation}
which is in fact the multi-species version of a classical formula~\cite[equation (9)]{ament_sound_1953} (in the absence of viscosity) and many others~\cite{martin_estimating_2010}. When performing this same procedure for the plane-wave dispersion equation~\eqref{eqn:planewave_eigensystem} or for azimuthal symmetry~\eqref{eqn:dispersion-azimuthal} we recover the same effective properties. Also note, that for a single species, with the properties $\beta_1$ and $\rho_1$, we recover the correct limits: when $\varphi \to 0$ we get $\beta_* \to \beta$ and $\rho_* \to \rho$, and when $\varphi \to 1$ we get $\beta_* \to \beta_1$ and $\rho_* \to \rho_1$.

\subsection{Numerical effective wavenumbers}
Here we numerically explore the effective wavenumbers which solve the dispersion equations: \eqref{eqn:regular-dispersion} with azimuthal symmetry~\eqref{eqn:f-azimuthal}, planar symmetry~\eqref{eqn:planewave_det}, and the combined planar with azimuthal symmetry~\eqref{eqn:azi-planewave_det}. The material properties we use are shown in the Table~\ref{tab:particle-properties}.
\begin{table}[h!]
\centering
\begin{tabular}{ l l}
  $\varphi = 30\%$ & \text{(particle volume fraction)}
\\
 $ka_o = \pi/8$ & \text{(non-dimensional particle radius)}
\\
 $\rho_o/\rho = c_o/c = 0.1$ & \text{(void particle properties)}
\\
$\rho_o/\rho = c_o/c = 10.0$ & \text{(solid particle properties)}
\\
\phantom{hrule}
\end{tabular}
\caption{For numerical experiments  we use identical spherical particles with the T-matrix~\eqref{eqn:T-matrix-acoustics} and the material properties given above, where $c = \omega / k$  and $c_o = \omega / k_o$ are the background and particle wave speed, respectively. Note that we give the properties of the particles relative to the background properties. The properties of the void particle are an example of strong scatterers, while that of the solid particle are an example of weak scatterers.
}
\label{tab:particle-properties}
\end{table}

\begin{optionalnote}[label={alg:wavenumbers}]{Algorithm - effective wavenumbers}
A minimal algorithm to calculate the effective wavenumbers $k_p$ for any material region.
\begin{enumerate}
    \item Choose the particle statistics by choosing:
    \begin{enumerate}[label*=\arabic*.]
        \item a wavenumber $k$, where $\omega = c k$ and $c$ is the background wave speed.
        \item the coefficients of the T-matrix $T_n$. One example is given by~\eqref{eqn:T-matrix-acoustics}.
        \item the function $p(\lambda_1)$ detailed in Section~\ref{sec:statistics}. For just one species $p(\lambda_1) = \delta(\lambda_1 - \lambda^*)$.
        \item the exclusion distance $a_{12}$, with $a_{12} = 1.001 (a_1 + a_2)$ being a common choice.
    \end{enumerate}
    \item Choose a truncation $\ell \leq L$, for the $\ell$ in $n = (\ell,m)$ and in $F_{p,nn_1}$ or $F_{p,n}$ based on how the $T_n$ decay.
    \item Calculate the effective wavenumbers $k_p$ by solving~\eqref{eqn:azi-planewave_det} or \eqref{eqn:planewave_det}.
\end{enumerate}

\end{optionalnote}

In Figures~\ref{fig:wavenumbers-compare}-\ref{fig:weak-wavenumbers} we show the result of using Algorithm~\ref{alg:wavenumbers} above to calculate the different effective wavenumbers when using the properties in Table~\ref{tab:particle-properties}. In both cases, the wavelength $\lambda$ is sixteen times larger than the particle radius. The important messages to take-away from these figure are:
\begin{enumerate}
    \item The more general regular dispersion equation has spurious roots, which are the ones that do not satisfy the planar dispersion, as discussed in Section~\ref{sec:plane-wave-disp-all}.
    \item There can be two effective wavenumbers with lower imaginary part (and are not spurious roots), as shown in Figure~\ref{fig:wavenumbers-compare}. These two will dominate calculate the ensemble average transmission and scattering, as the other wavenumbers will be very difficult to excite. For weaker scatterers there tends to be only one wavenumber with lower imaginary part as shown in Fgiure~\ref{fig:weak-wavenumbers}.
    \item The simpler combined planar-azimuthal dispersion~\eqref{eqn:azi-planewave_det} equation contains the two most important wavenumbers. This seems to hold in general.
\end{enumerate}
We remark that in most cases we find that there is only one wavenumber with a low imaginary part, and knowing this wavenumber is often enough to accurately calculate the ensemble average transmission and scattering. This is exactly what we do in the next sections.


\begin{figure}[ht!]
    \centering
    \includegraphics[width=0.7\linewidth]{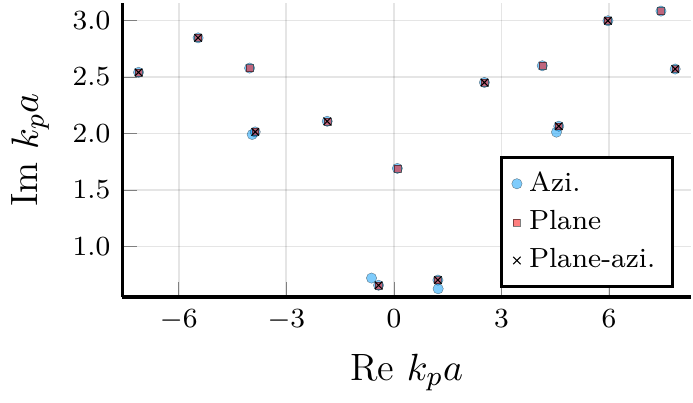}
    \caption{The effective wavenumbers that satisfy azimuthal symmetry~\eqref{eqn:dispersion-azimuthal}, the planar dispersion~\eqref{eqn:planewave_det}, or the combined planar with azimuthal dispersion~\eqref{eqn:azi-planewave_det}. We used the material properties in Table~\ref{tab:particle-properties} for void particles and incident wavenumber times particle radius $k a_o  = \pi / 8$, which means the wavelength is 16 times longer than the particle radius. }
    \label{fig:wavenumbers-compare}
\end{figure}

\begin{figure}[ht!]
    \centering
    \includegraphics[width=0.7\linewidth]{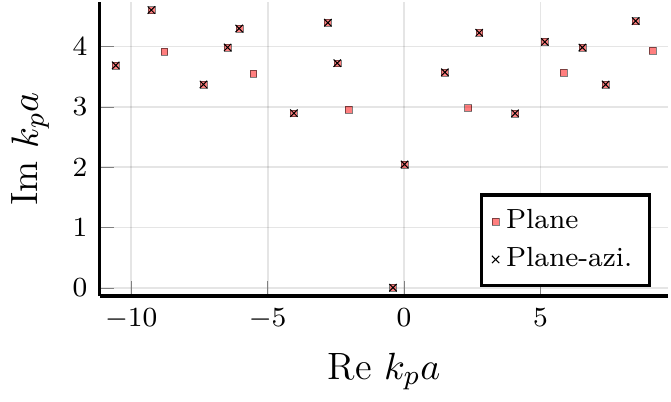}
    \caption{The effective wavenumbers that satisfy the planar dispersion~\eqref{eqn:planewave_det} or the combined planar with azimuthal dispersion~\eqref{eqn:azi-planewave_det}. Note that the wavenumber closest to the $x$-axis will be the most important, as it has a smaller imaginary part. We used the material properties in Table~\ref{tab:particle-properties} for solid particles and incident wavenumber times particle radius $k a_o  = \pi / 8$. }
    \label{fig:weak-wavenumbers}
\end{figure}

\section{Scattering from a sphere filled with particles}

If one effective wavenumber $k_1$ has a significantly smaller imaginary part than the other wavenumbers, e.g.\ Im $k_1 \ll$ Im $k_p$ for $p=2,3, \ldots$, then $\ensem{f_n}(\rv_1,\lambda_1) \approx f_{1,n}(\rv_1,\lambda_1)$, where $f_{1,n}(\rv_1,\lambda_1)$ is the wavemode associated with $k_1$ as shown in the representation~\eqref{eqn:effective_wave_representation}. This occurs for a number of scenarios including: weak scattering, low frequency, or low volume fraction.
In this case, we can explicitly calculate the average scattered and transmitted waves for many different material geometries. To achieve this, for each material geometry, we specialise the average scattered wave~\eqref{eqn:average-field} to the material geometry, then use the average boundary conditions~\eqref{eqn:general_boundary} to restrict the wavemode $f_{1,n}(\rv_1,\lambda_1)$. With the wavemode we can calculate both the average scattered and transmitted wave, although transmission requires some extra steps~\cite{martin_multiple_2011}.

In this section, we calculate the average scattered wave from a sphere filled with particles. To our knowledge, the sphere case has never been analytically calculated in all its details, though there have been approximate methods~\cite{muinonen_coherent_2012} and numerical methods that simulate a large number of configurations~\cite{Mackowski+Mishchenko2011a,Mackowski+Mishchenko2011,Mackowski2001}.

\subsection{The average boundary conditions}
We assume that all particles are confined in a sphere of radius $R$ which implies that the particle origins  $\reg_1 = \{\rv_1 \in \R^3: |\rv_1|\leq R - a_1\}$, and let the centre of the sphere  be the origin of the coordinate system for $\rv_1$.
Choosing a simple geometry  allows us to explicitly calculate the average boundary conditions \eqref{eqn:Ireg}. Assume that the $k_p$ have been determined from~\eqref{eqn:planewave_det} and that the $F_{p,n n_1}$, up to a multiplying constant $\alpha_p$, have been determined from~\eqref{eqn:regular-dispersion}. The results below can be used to completely  determine, or just restrict, the $\alpha_p$.

Let $\bar R_2 = R - a_2$, then

\begin{multline}
     \mathcal I_p (\rv_1) = \sum_{n_1}\alpha_p F_{p,n'n_1}(\lambda_2)
     \int_{r = \bar R_2}  \mathcal{U}_{n'n}(k\rv_1 - k\rv_2)
     \frac{\partial   \mathrm v_{n_1} (k_p\rv_2)}{\partial \boldsymbol\nu_2}
     - \frac{\partial \mathcal{U}_{n'n}(k\rv_1 - k\rv_2) }{\partial \boldsymbol\nu_2}  \mathrm v_{n_1} (k_p\rv_2) \mathrm  dA_2\\
     = - \bar R_2 \sum_{n_1 n_2} \alpha_p F_{p,n'n_1}(\lambda_2) c_{n'n n_2}\mathcal{V}_{n_2,(l_1,-m_1)}(k\rv_1)
    {\mathrm N}_{l_1}(k\bar R_2,k_p\bar R_2)(-1)^{l_1+m_1}
    \\
      = \sum_{n_1n_2}\bar R_2 {\mathrm N}_{l_1}(k\bar R_2, k_p\bar R_2) \alpha_p F_{p,n'n_1}(\lambda_2)B_{n n_2, n' n_1} \mathrm v_{n_2} (k\rv_1) \quad \text{for}\;\; |\rv_1| \leq \bar R_2 - a_{12},
\end{multline}
where we used the translation matrices~\eqref{eq:Translation matrix-scalar} to write $\mathcal{U}_{n'n}(k \rv_1-k \rv_2) = \sum_{n_1} c_{n' n n_1} \mathrm u_{n_1}(k \rv_1-k \rv_2) = \sum_{n_1 n_2} c_{n' n n_1} \mathcal{V}_{n_1 n_2}(k \rv_1) \mathrm u_{n_2}(-k \rv_2)$, which relied on $|\rv_1| < \bar R_2 - a_{12} < |\rv_2| = \bar R_2$, we then used the orthogonality of the spherical harmonics to resolve the integral, defined
\begin{equation}
     B_{n n_2, n' n_1}= - \sum_{n_3}c_{n'nn_3}c_{n_2n_1n_3},
\end{equation}
and used $c_{n_3(\ell_1,-m_1)n_2} = c_{n_2n_1n_3} (-1)^{\ell_1 + m_1}$.
Substituting the above into~\eqref{eqn:general_boundary} leads to
\begin{equation} \label{eqn:ExtinctionSphere}
\tcboxmath{
         \mathbf B  \sum_p \alpha_p \vec F_p =  \mathbf L  \vec g,
        } \quad \text{(boundary conditions for a sphere)}
\end{equation}
where the components of $\vec g$ are $g_n$ from the expansion of the incident wave~\eqref{eqn:incident-spherical}, the $\vec F_p$ are vectors with components
\begin{equation}
    (\vec F_p)_{n'n_1} =  \int_{\mathcal S}\frac{R - a_2}{k^2 - k^2_p} {\mathrm N}_{l_1}(k R - k a_2,k_pR-k_pa_2) F_{p,n'n_1}(\lambda_2)  \numdensity(\lambda_2) \mathrm d \lambda_2,
\end{equation}
and again we used the translation matrices \eqref{eq:Translation matrix-scalar} and the orthogonality of the spherical harmonics. The vector
$\vec a$ has the components $(\vec a)_{n'} = a_{n'}$, and
$\mathbf L$ and $\mathbf B$ are matrices with components
\begin{equation}
    (\mathbf L)_{n n_2, n'} = - c_{n'nn_2} \quad \text{and} \quad
    (\mathbf B)_{n n_2, n' n_1}  =  - \sum_{n_3}c_{n'nn_3}c_{n_2n_1n_3}.
\end{equation}

\begin{optionalnote}{Sparse boundary conditions}
We find that numerically~\eqref{eqn:ExtinctionSphere} has a unique solution for the $\alpha_p$ when using only one effective wavenumber $k_1$. This occurs even though~\eqref{eqn:ExtinctionSphere} often has far more equations than unknowns $\alpha_p$. That is, the number of equations is larger than the number of effective eigenvectors $P$.
Nonetheless, both the left and right-hand side are sparse, and mostly filled with zeros. This indicates that it is possible to transform~\eqref{eqn:ExtinctionSphere} into a smaller equivalent system to determine the $\alpha_p$.
\end{optionalnote}

\subsection{Average scattered field}
Assume the particles are confined in a spherical region of radius $R$. For the particles to fit in this region the particle origins need to be contained within $\reg_1 = \{\rv_1 \in \R^3: |\rv_1|\leq R - a_1\}$, where we let the centre of the sphere be the origin of the coordinate system for $\rv_1$.
In this case, by taking $|\rv| > R$ we can use the wave representation~\eqref{eqn:effective_wave_representation}, Green's second identity, and then~\eqref{eqns:fields_spherical_basis}, to reduce the average scattered wave~\eqref{eqn:average-field}:
\begin{align*}
    & \sum_{n p} \int_{\reg_1} f_{p,n}(\rv_1,\lambda_1) \mathrm u_n (k  \rv - k \rv_1) \mathrm d \rv_1
    \\
    & =  \sum_{n p} \frac{1}{k^2 -k^2_p} \int_{\partial\reg_1}
    \left [ \frac{\partial f_{p,n}(\rv_1,\lambda_1)}{\partial r_1} \mathrm u_n (k  \rv - k \rv_1) -  f_{p,n}(\rv_1,\lambda_1) \frac{\partial \mathrm u_n (k  \rv - k \rv_1) }{\partial r_1} \right] (R-a_1)^2  \mathrm d \Omega_1
    \\
      & =  \sum_{n'}  \mathrm  u_{n'}(k\rv)  \mathfrak F_{n'}(\lambda_1)
    \end{align*}
where we used~(\ref{eq:translation_spherical_waves}-\ref{eq:Translation matrix-scalar}) to write $\mathrm u_n(k \rv - k \rv_1) = \sum_{n' n_2} c_{n'nn_2}\mathrm v_{n_2}^*(k^* \rv_1) \mathrm u_{n'}(k\rv)$ followed by the orthogonality of the spherical harmonics, and we defined
\begin{equation} \label{eqn:region-scattering-coefficients}
    \mathfrak F_{n'}(\lambda_1) =  \sum_{n p n_1} \frac{R-a_1}{k_p^2 -k^2}  c_{n'nn_1} F_{p,n n_1}(\lambda_1)
    \mathrm M_{\ell_1} (k (R-a_1), k_p (R-a_1)),
\end{equation}
 and $\mathrm M_{\ell} (x, y) = x \mathrm j_{\ell}'(x) \mathrm j_{\ell}(y) - y \mathrm j_{\ell}(x) \mathrm j_{\ell}'(y)$.
The $\mathfrak F_{n'}(\lambda_1)$ are the scattering coefficients of the whole spherical region~\eqref{eqn:total-scattered-field}. In conclusion, substituting the above into~\eqref{eqn:average-field} leads to
\begin{equation} \label{eqn:average-field-sphere}
\tcboxmath{
    \ensem{\us(\rv)} =  \sum_{n'} \mathrm u_{n'}(k \rv) \int_{\mathcal S} \mathfrak F_{n'}(\lambda_1) \numdensity(\lambda_1) d \lambda_1 , \quad r> R,
    } \quad \text{(average scattered field)}
\end{equation}
where we averaged over particle rotations.

To help piece together the equations, we provide Algorithm below to calculate the average scattered wave above when only one effective wavenumber $k_1$ has a smaller imaginary part than the others as shown in Figure~\ref{fig:weak-wavenumbers}.

\begin{optionalnote}[label={alg:sphere-scattering}]{Algorithm - scattering from a sphere filled with particles.}
This algorithm assumes one effective wavenumber has a imaginary part significantly smaller than the others.
\begin{enumerate}
    \item Use Algorithm~\ref{alg:wavenumbers} to calculate the wavenumber $k_1$ with smallest imaginary part.
    \item \label{it:L1} Choose a truncation $\ell_1 \leq L_1$ for the $n_1$ in $F_{p,nn_1}$ based on how the $M_{\ell_1}$ in~\eqref{eqn:total-scattered-field} decay. The range of all other indices can now be determined from these truncations and the properties~\eqref{eqns:gaunt_indices}.
    \item Calculate the multiple eigenvectors $F_{p,nn_1}(\lambda_1)$ of $k_1$ by solving~\eqref{eqn:regular-dispersion}. From~\eqref{eqn:effective_wave_representation} and \eqref{eqns:fields_spherical_basis}
    we now have that $\ensem{f_n}(\rv_1,\lambda_1) = \sum_{n_1} \mathrm v_{n_1} (k_1\rv_1) \sum_{p=1}^P \alpha_p F_{p,nn_1}(\lambda_1)$, where the $\alpha_p$ need to be determined. \item Solve \eqref{eqn:ExtinctionSphere} by:
    \begin{enumerate}[label*=\arabic*.]
    \item setting\footnote{For azimuthal symmetry this becomes $L' = P - 1$} $L' = \sqrt{P} - 1$ where $\ell' \leq L'$ in  $g_{n'}$. This causes the number of unknowns $\alpha_p$ to be equal to the number of coefficients $g_n'$, which leads to a unique solution $\alpha_p$.
    \item setting $\mathbf A = \mathbf B \mathbf F$,  where $\mathbf F \vec \alpha = \sum_{p=1}^P  \mathbf F_p \alpha_p$, then $\vec \alpha = \mathbf A^+ \mathbf L  \vec g$ where $(\vec \alpha)_p = \alpha_p$
    and $\mathbf A^+ = ((\mathbf A^T)^* \mathbf A)^{-1} (\mathbf A^T)^*$ is the pseudo-inverse.
    \end{enumerate}
    \item Finally, calculate the scattering coefficients of the whole spherical region~\eqref{eqn:region-scattering-coefficients}.
\end{enumerate}


{\bf Restriction to azimuthal symmetry.} Substitute: $F_{p,n'n_1} = \delta_{m_1,-m'} \delta_{\ell_1 \geq |m'|} F_{p,n'\ell_1}$, to satisfy~\eqref{eqn:azimuthal-symmetry}, and $a_{n'} = \delta_{m',0} a_{\ell'}$, to make the incident wave satisfy azimuthal symmetry. For an incident plane wave impinging in the $\unitvec{k}$ direction $a_{n'} = 4 \pi \iu^{\ell'} Y_{n'}^*(\unitvec{k})$. When solving \eqref{eqn:ExtinctionSphere} take $m_2 = -m$ and only evaluate $\ell_2 \geq |m|$.
\end{optionalnote}

\section{A plate filled with particles}

In this section, we calculate the average reflected wave from a plate region $\reg=\{\rv\in\R^3: Z_1\leq z \leq Z_2\}$ filled with particles, and the average transmitted wave that passes through to the other side of the plate. In this case, the particle origins are confined to the region $\reg_1 = \{\rv_1 \in\R^3: Z_1 + a_1 \leq z_1 \leq Z_2 - a_1\}$, where $Z_2$ needs to be large enough so that $ Z_1 + a_1 + a_{12} < Z_2 - a_1 - a_{12}$ for every particle radius $a_1$ and minimum inter-particle distance $a_{12}$.
We now assume that the incident wave is a plane wave with wave-vector $\vec k = (k_x,k_y,k_z)$, where $k_z>0$ (the incident wave impinges the plate from below). Further, by planar symmetry~\eqref{eqn:plane-wave-symmetry} we can also assume that ${k_p}_x = k_x$ and ${k_p}_y = k_y$.

\subsection{Average transmission}
We start with the transmitted field using the wave representation~\eqref{eqn:total-field}, \eqref{eqn:average-field}, \eqref{eqn:effective_wave_representation}, and~\eqref{eqn:effective-plane-wave}. By assuming that $z > Z_2$, which is the side of the plate where the transmitted wave will appear, we can use Green's second identity, as we did in Section~\ref{sec:wave-decomposition}, to reduce the average (total) transmitted wave:
\begin{equation} \label{eqn:average-transmitted field}
    \ensem{u(\rv)}=\eu^{\iu \vec{k}\cdot\rv}+\ensem{\us(\rv)} =
    \eu^{\iu \vec{k}\cdot\rv} + \sum_{n,p} \int_{\mathcal S} \numdensity (\lambda_1) F_{p,n}(\lambda_1)
    \int_{\reg_1} \eu^{\iu\vec{k}_p\cdot\rv_1} \mathrm u_n (k \rv - k \rv_1) \,\diff\rv_1 \diff\lambda_1.
\end{equation}
By using Green's second identity we can reduce the integral in $\rv_1$ to surface integrals:
\begin{equation} \label{eqn:greens-identity-plate}
    \int_{\reg_1}
     \eu^{\iu \vec k_p \cdot \rv_1} \mathrm u_n (k \rv - k \rv_1)  \mathrm d \rv_1 =
     \frac{\eu^{\iu \vec k_p \cdot \vec x}}{k^2 - k_p^2} (-1)^\ell \left [ L_n(Z_2 - a_1 -z) - L_n(Z_1 + a_1 - z) \right ],
\end{equation}
where we used a change of integration variable from $\rv_1 \mapsto (\rv_1 - \rv)$,  $\mathrm u_n ( - k  \rv_1) = \mathrm u_n (k  \rv_1) (-1)^\ell $, and the definition of $L_n$ which is given by~\eqref{def:L}.
As both $Z_2 - a_1 -z$ and $Z_1 + a_1 -z$ are negative real numbers, we can use the result~\eqref{eqn:L-integrated-symmetry} to evaluate the above and reach \begin{equation} \label{eqn:average-transmitted field-2}
\tcboxmath{
    \ensem{u(\rv)}=
    T_\text{plate} \eu^{\iu \vec{k}\cdot\rv},
} \quad \text{(average transmitted field)}
\end{equation}
where
\[
T_\text{plate} = 1 + \sum_{n,p} \int_{\mathcal S} \numdensity (\lambda_1) F_{p,n}(\lambda_1)
    \frac{(-1)^\ell}{k_z - k_{p_z}}  \frac{2 \pi  \iu^{\ell + 1}}{ k k_z}   \mathrm Y_n(\unitvec{k})
    \left [ \eu^{\iu (k_{p_z} - k_z)(Z_2 - a_1) }  - \eu^{\iu (k_{p_z} - k_z)(Z_1 + a_1) }  \right ] \diff\lambda_1,
\]
and $\vec k_p = (k_x,k_y,k_{p_z})$ due to planar symmetry~\eqref{eqn:plane-wave-symmetry}, for some complex number $k_{p_z}$, and we used $k^2 - k_p^2 = k_z^2 - k_{p_z}^2$.

\subsection{Average reflection}
The average scattered wave is again given by~\eqref{eqn:average-field}, which  by substituting the effective wave representation
\eqref{eqn:effective_wave_representation} and \eqref{eqn:effective-plane-wave}  leads to
\begin{equation} \label{eqn:average-reflected}
    \ensem{\us(\rv)} =
     \sum_{n,p}  \int_{\mathcal S}  \numdensity(\lambda_1) F_{p,n}(\lambda_1) \int_{\reg_1}
     \eu^{\iu \vec k_p \cdot \rv_1} \mathrm u_n (k \rv - k \rv_1)  \mathrm d \rv_1 \mathrm d  \lambda_1, \quad \text{for} \quad z < Z_1.
\end{equation}
By using Green's second identity we can reduce the integral over $\reg_1$ as done in~\eqref{eqn:greens-identity-plate}. As the plate is thicker than any one particle, we then have that $Z_2 - a_1 - z >0$ and $Z_1 +a_1 -z >0$ for $z < Z_1$, which allows us to pick the positive argument in~\eqref{eqn:L-integrated}, evaluate~\eqref{eqn:greens-identity-plate}, and reduce~\eqref{eqn:average-reflected} to
\begin{equation}
\tcboxmath{
\ensem{\us(\rv)} =
     R_\text{plate}\eu^{\iu \vec k_{\mathrm{ref}} \cdot \vec x},
     } \quad \text{(average reflected field)}
\end{equation}
where $R_\text{plate}$ is the reflection coefficient:
\begin{equation}
R_\text{plate} =
     \sum_{n,p} \frac{2 \pi \iu^{\ell-1}}{k_z + k_{p_z}}  \frac{\mathrm Y_n(\unitvec k_{\mathrm{ref}})}{k k_z}  \int_{\mathcal S}  \numdensity(\lambda_1) F_{p,n}(\lambda_1)  \left [
      \eu^{\iu (k_{p_z} + k_z) (Z_2 - a_1)}
     - \eu^{\iu (k_{p_z} + k_z) (Z_1 + a_1)} \right ] \mathrm d  \lambda_1,
\end{equation}
we define $\vec k_{\mathrm{ref}} = (k_x,k_y,-k_z)$, used $\vec k_p = (k_x,k_y,k_{p_z})$ due to planar symmetry~\eqref{eqn:plane-wave-symmetry}, and that $k^2 - k_p^2 = k_z^2 - k_{p_z}^2$.

\subsection{The average boundary conditions}
When using only one effective wavenumber $k_1$ the equation~\eqref{eqn:general_boundary} can be used to fully determine the field~\eqref{eqn:effective-plane-wave}, like a boundary condition. Note we can use planar symmetry~\eqref{eqn:planar-symmetry} and the form~\eqref{eqn:effective-plane-wave} because both the incident wave and material region share a planar symmetry.

The first step is to simplify \eqref{eqn:Ireg}:
\begin{multline}
    \mathcal I_p (\rv_1) = F_{p,n'}(\lambda_2)\int_{z_2=Z_2-a_2} \mathcal{U}_{n'n}(k\rv_1 - k\rv_2) \frac{\partial \eu^{\iu\vec{k}_p\cdot\rv_2}}{\partial z_2}
     - \frac{\partial \mathcal{U}_{n'n}(k\rv_1 - k\rv_2) }{\partial z_2} \eu^{\iu\vec{k}_p\cdot\rv_2} \,\diff x_2\diff y_2
     \\
      -F_{p,n'}(\lambda_2)\int_{z_2=Z_1+a_2}  \mathcal{U}_{n'n}(k\rv_1 - k\rv_2) \frac{\partial \eu^{\iu\vec{k}_p\cdot\rv_2}}{\partial z_2}
     - \frac{\partial \mathcal{U}_{n'n}(k\rv_1 - k\rv_2) }{\partial z_2} \eu^{\iu\vec{k}_p\cdot\rv_2} \,\diff x_2\diff y_2,\quad\rv_1\in\reg_1(a_{12})
     \end{multline}
To explicitly calculate the above integrals, we use the translation matrices in Appendix~\ref{sec:Translation}, followed by changing the integration variable to $\rv=\rv_2-\rv_1$ and then using the definition~\eqref{def:L} to obtain
\begin{equation} \label{eqn:Ip-rewrite}
    \mathcal I_p (\rv_1) = F_{p,n'}(\lambda_2)\eu^{\iu\vec{k}_p\cdot\rv_1}\sum_{n_1}(-1)^{\ell_1}c_{n'nn_1} \left[L_{n_1}(Z_2 - a_2 - z_1)
     -  L_{n_1}(Z_1 + a_2 - z_1) \right],
\end{equation}
where
factor $(-1)^{\ell_1}$ appeared when substituting $\mathrm u_{n_1}(k\rv_1 -k\rv_2) = (-1)^{\ell_1} \mathrm u_{n_1}(k\rv_2 -k\rv_1)$ in the  integrals.

We can use the formula~\eqref{eqn:L-integrated-symmetry} to easily calculate $L_{n_1}$ by noting that $Z_2 - a_2 - z_1 > 0$ and $Z_1 + a_2 - z_1 < 0$. These inequalities are a result of $\rv_1\in\reg_1(a_{12})$ which implies that $Z_1 + a_1 + a_{12} \leq z_1 \leq Z_2 - a_1 - a_{12}$. Substituting~\eqref{eqn:L-integrated-symmetry} into \eqref{eqn:Ip-rewrite} then leads to
\begin{multline} \label{eqn:Ip-plate}
    \mathcal I_p (\rv_1) = F_{p,n'}(\lambda_2)\eu^{\iu\vec{k}_p\cdot\rv_1}\sum_{n_1}\iu^{\ell_1} c_{n_1 (\ell,-m) n'} \mathrm{Y}_{n_1}(\unitvec{k})  \times
    \\
    \frac{2 \pi \iu}{k k_z} (-1)^{\ell +m}  \left[ (-1)^{m' - m}({k_p}_z-k_z)
    \eu^{\iu ({k_p}_z + k_z) (Z_2 - a_2 -z_1)}
     -  (-1)^{\ell' + \ell} ({k_p}_z + k_z)
    \eu^{\iu ({k_p}_z - k_z) (Z_1 + a_2 - z_1)} \right].
\end{multline}
where we replaced $(-1)^{m_1} =(-1)^{m' - m}$, $(-1)^{\ell_1} = (-1)^{\ell' + \ell}$, and $c_{n'nn_1} = (-1)^{\ell +m} c_{n_1 (\ell,-m) n'}$ by using the properties of $c_{n'nn_1}$ shown in Appendix~\ref{sec:Translation}.
These replacements allow us to simplify~\eqref{eqn:Ip-plate} by applying the contraction rule~\eqref{eqn:spherical-linearisation} and some rearrangement to reach:
\begin{equation} \label{eqn:Ip-plate-2}
    \mathcal I_p (\rv_1) = 4 \pi \iu^\ell \mathrm Y_{(\ell,-m)}(\unitvec{k}) \left [ (-1)^{\ell} I_2  \eu^{\iu (k_x,k_y,-k_z) \cdot \vec r_1} + (-1)^m I_1  \eu^{\iu \vec k \cdot \vec r_1} \right],
\end{equation}
where
\begin{equation} \label{eqn:I_T-I_R}
   \left\{\begin{aligned}
     & I_1 =  - F_{p,n'}(\lambda_2)
     \mathrm Y_{n'}(\unitvec{k})  \iu^{\ell'}  \frac{2 \pi \iu}{k k_z}  (-1)^{\ell'} ({k_p}_z + k_z)
     \eu^{\iu ({k_p}_z - k_z) (Z_1 + a_2 )},
     \\
   & I_2 =  F_{p,n'}(\lambda_2)
     \mathrm Y_{n'}(\unitvec{k})  \iu^{\ell'}  \frac{2 \pi \iu}{k k_z} (-1)^{ m'}({k_p}_z-k_z)
     \eu^{\iu ({k_p}_z + k_z) (Z_2 - a_2)},
\end{aligned}\right.
\end{equation}
and we also used ${k_p}_x = k_x$ and ${k_p}_y = k_y$.
Substituting the above into~\eqref{eqn:general_boundary} leads to
\begin{equation} \label{eqn:plate-boundary-conditions}
    (-1)^{m} \eu^{\iu \vec k \cdot \vec r_1} \left[ 1  +
        \sum_{n'p}\int_{\mathcal S} \frac{I_1  \bar  \numdensity(\lambda_2)}{k^2 - k^2_p}
     \mathrm d \lambda_2
     \right]
        +
        \eu^{\iu (k_x,k_y,-k_z) \cdot \vec r_1} (-1)^{\ell} \sum_{n'p}\int_{\mathcal S} \frac{I_2 \bar  \numdensity(\lambda_2)}{k^2 - k^2_p} \mathrm d \lambda_2
         = 0,
\end{equation}
where we used that
\[
\sum_{n'}\mathcal{V}_{n'n}(k\rv_1)g_{n'} = 4 \pi \iu^{\ell} (-1)^{m} \eu^{\iu \vec k \cdot \vec r_1}  \mathrm Y_{\ell -m}(\unitvec k),
\]
which holds for incident plane waves with the coefficients~\eqref{eq:a_n def} and can be shown by using {\eqref{eqn:plane-wave-expansion}} and the contraction rule~\eqref{eqn:spherical-linearisation}.

For {\eqref{eqn:plate-boundary-conditions}} to hold for every $\vec r_1 \in \reg_1(a_{12})$ leads to two equations: one for the term multiplying $\eu^{\iu (k_x,k_y,-k_z) \cdot \vec r_1}$ and another for the terms multiplying $\eu^{\iu \vec k \cdot \vec r_1}$. These two equations can be written in the form:
\begin{align}
&  \label{eqn:reflection-boundary}
\tcboxmath{
     1  + \sum_{n'p}\int_{\mathcal S} \frac{I_1  \bar  \numdensity(\lambda_2)}{k^2 - k^2_p} \mathrm d \lambda_2
         = 0,
}
\quad \text{(surface $Z_1$)}
\\
& \label{eqn:transmission-boundary}
\tcboxmath{
    \sum_{n'p}\int_{\mathcal S} \frac{I_2 \bar  \numdensity(\lambda_2)}{k^2 - k^2_p}  \mathrm d \lambda_2 = 0,
}
\quad \text{(surface $Z_2$)}
\end{align}
where $I_1$ and $I_2$ are given in~\eqref{eqn:I_T-I_R}.

For a finite plate, both \eqref{eqn:reflection-boundary} and \eqref{eqn:transmission-boundary} need to be enforced to restrict the $F_{p,n}(\lambda_1)$. If the sum over $p$ has only two terms, using one a forward propagating and the other a backward propagating mode, then these equations can be used to obtain a unique solution for the $F_{p,n}$. This typically occurs when using only one effective wavenumber $k_1$. For reflection from a halfspace, only \eqref{eqn:reflection-boundary} should be enforced, which is the multi-species three dimensional version of \cite[Equation (20)]{martin_multiple_2011}.

\section{Numerical results: plane-wave incident on a particulate sphere}
\label{sec:numerical-results}

This paper is the first, to our knowledge, to provide analytic solutions for the average wave scattered from particles within a spherical region $\reg$ as given by \eqref{eqn:ExtinctionSphere} and \eqref{eqn:average-field-sphere}. The methods used previously\cite{Mackowski2001,muinonen_coherent_2012} have  approximated the scattered field by assuming that the ensemble averaged sphere behaves like a homogeneous sphere occupying the region $\reg$ with some effective properties.
For these reasons, in this section we numerically compare these approaches.

For all the results below, we avoid combining a high particle volume fraction with a high frequency, as this regime triggers multiple wavenumbers with low imaginary parts, as shown in Figure~\ref{fig:wavenumbers-compare}.
Whereas the calculations below rely on equation~\eqref{eqn:ExtinctionSphere} giving a unique solution, which only occurs when using just one effective wavenumber.
Using only one wavenumber is {an} excellent approximation when its imaginary part is much smaller than all the other wavenumbers, as shown in Figure~\ref{fig:weak-wavenumbers}.
See~\cite{gower2019proof,gower2019multiple} for details on how to calculate reflection and transmission when multiple effective wavenumbers have a small imaginary part.

\begin{figure}[ht]
    \centering
    \begin{tikzpicture}
        \draw (0, 0) node[inner sep=0] {\includegraphics[scale=0.45]{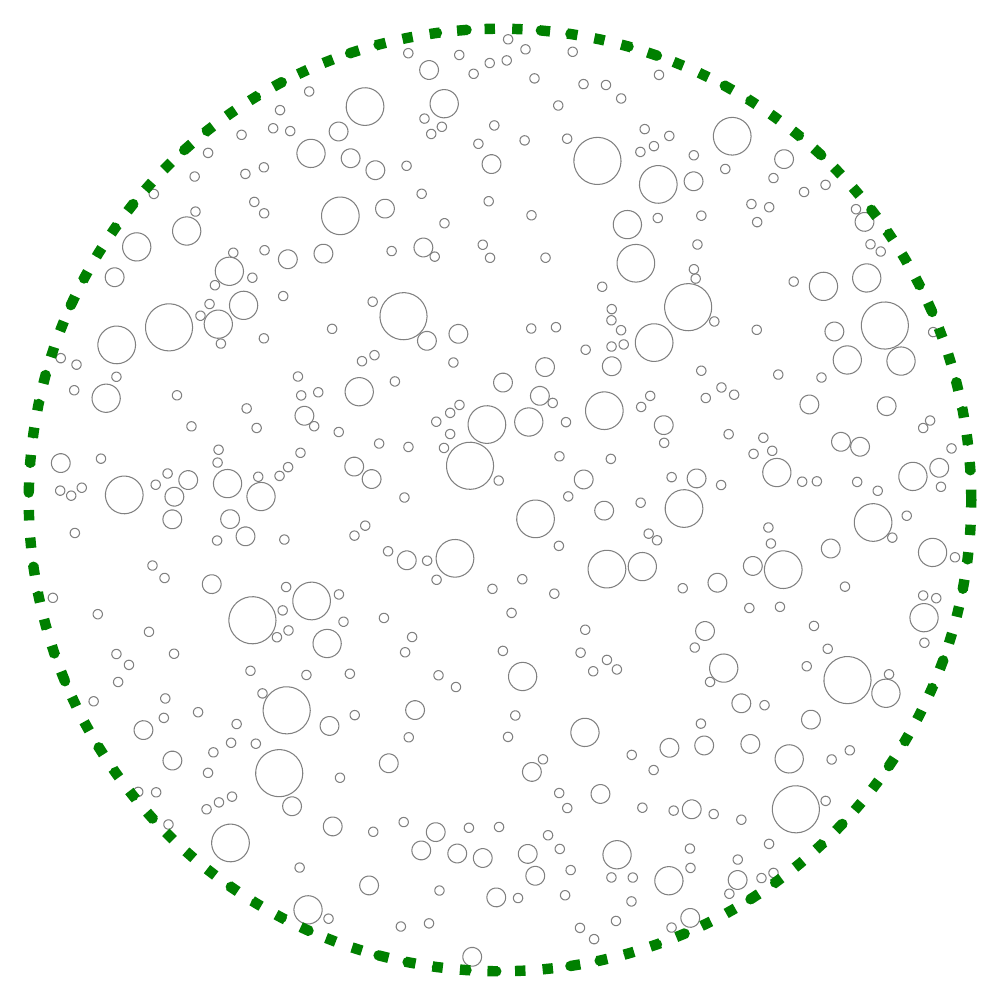}};
        \draw (0, 0) node {\huge $\reg$};
    \end{tikzpicture}
    \includegraphics[scale = 0.45]{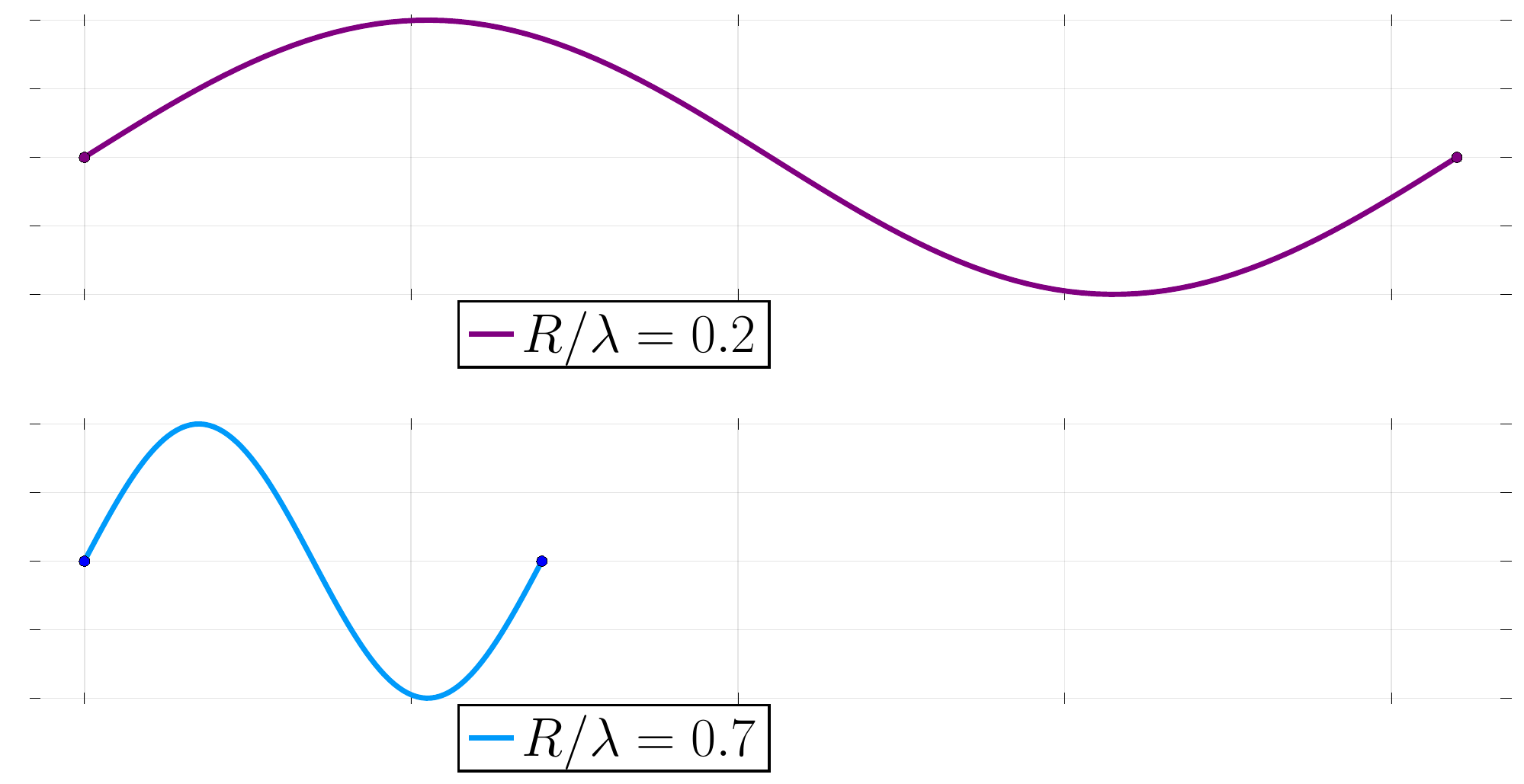}
   \caption{ An illustration of the size of the spherical region $\reg$, the size of particles inside, and the typical length of the incident wavelengths $\lambda$ used for the results in Figures~\ref{fig:cross-section} - 
    \ref{fig:cross-section-volfrac}. Note that $R$ is the radius of $\reg$, and numerical parameters used in this section are in Table~\ref{tab:numerical-properties}. }
    \label{fig:compare-wavelengths}
\end{figure}

\begin{table}[h]
\centering
\begin{tabular}{l| c  c  c c}
   & mass density & wave speed  & radius & volume \% \\ [0.5ex]
\hline
Solid particles & $\rho_s/\rho = 10$ & $c_s / c = 10$ & $a_s / R =  1 / 20$ & $\varphi_s = 15\%$\\
\hline
Void particles & $\rho_v /\rho = 0.1$ & $c_v / c = 0.1$ & $a_v / R = 1 / 20$ & $\varphi_v = 5\%$
\end{tabular}
\caption{Material properties used for the numerical experiments. Note that $R$ is the radius of the spherical region $\reg$.}
\label{tab:numerical-properties}
\end{table}

{\bf The analytic scattered field.} For each {(angular)} frequency $\omega$, we calculate the effective wavenumber $k_1$ with the smallest imaginary part{. We} then follow Algorithm~\ref{alg:wavenumbers} and  Algorithm~\ref{alg:sphere-scattering} to calculate the scattered field. As a reminder, this method does not assume the spherical region behaves like some homogeneous sphere; instead these results are from careful homogenisation of all the scattered waves.

{\bf Two {different} homogeneous spheres.}  We can approximately calculate the scattered wave from the sphere $\reg$ by assuming that $\reg$ is filled with some homogeneous material. Below, we choose two different ways to approximate the density and sound speed of this homogeneous material we use to fill $\reg$. Note that there are many possible chooses for the density and sound speed and no clear ``best choice''.
\begin{enumerate}
    \item {\bf Hom. Low Freq.}: we assume the sphere has the effective density $\rho_*$ and {effective bulk module $\beta_*$ given by~\eqref{eqns:effective-properties}, which results in the sound speed $c_*=\sqrt{\beta_*/\rho_*}$}.
    \item {\bf Hom. Complex $k_1$}: we assume the sphere has the same complex wavenumber $k_1$ used for the analytic solution, which then implies it has sound speed $c_1 = \omega / k_1$. For the effective density, we again choose $\rho_*$ given by~\eqref{eqns:effective-properties}.
\end{enumerate}
After choosing one of these approximations, we can calculate the scattering coefficients $F_n$ by using the T-matrix~\eqref{eqn:T-matrix-acoustics}. Taking the origin to be the centre of the spherical region $\reg$, we can then express the scattered field in the form
\[
    \text{(scattered field)} =  \sum_{n}F_n \mathrm u_{n}(k \rv), \quad r> R,
\]
for a sphere of radius $R$. For the analytic solution we have  $F_n=\int_{\mathcal S} \mathfrak F_{n}(\lambda_1) \numdensity(\lambda_1) d \lambda_1$ from~\eqref{eqn:average-field-sphere}, and for the numerical results we approximate the integral as a sum.

\begin{figure}[ht]
    \centering
    \includegraphics[width=0.98\linewidth]{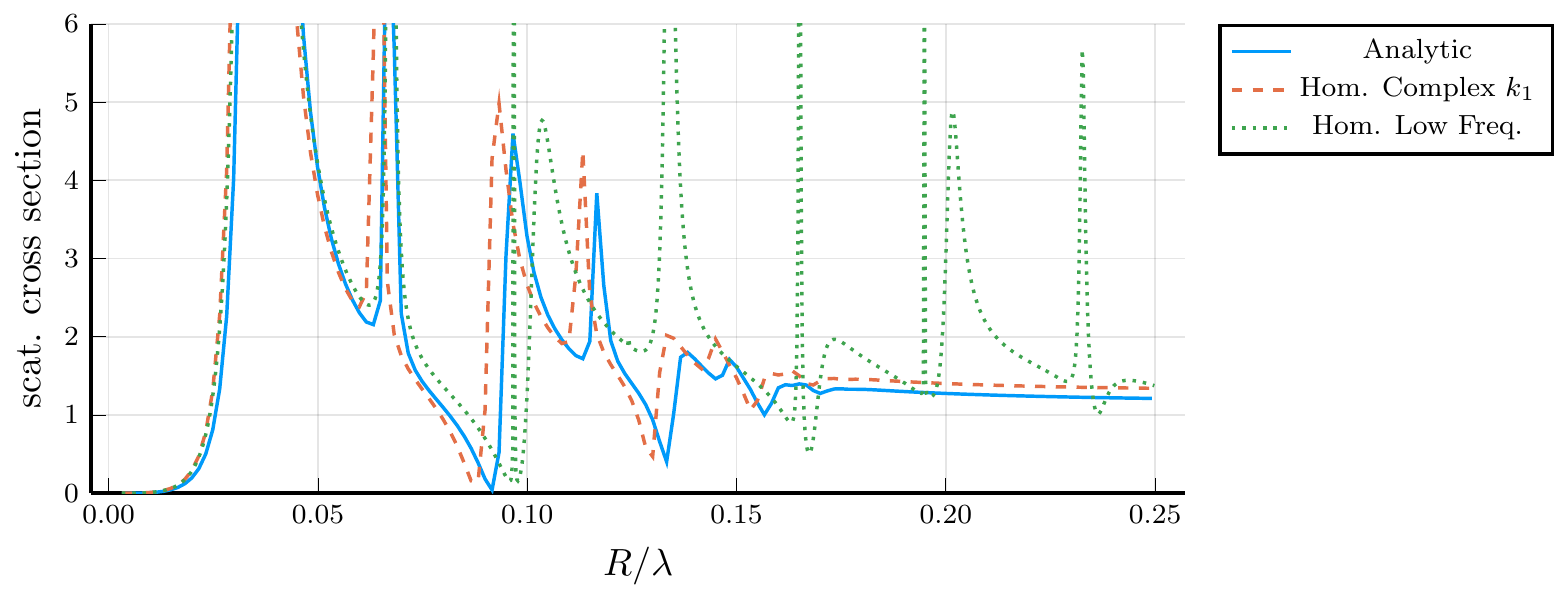}
    \includegraphics[width=0.98\linewidth]{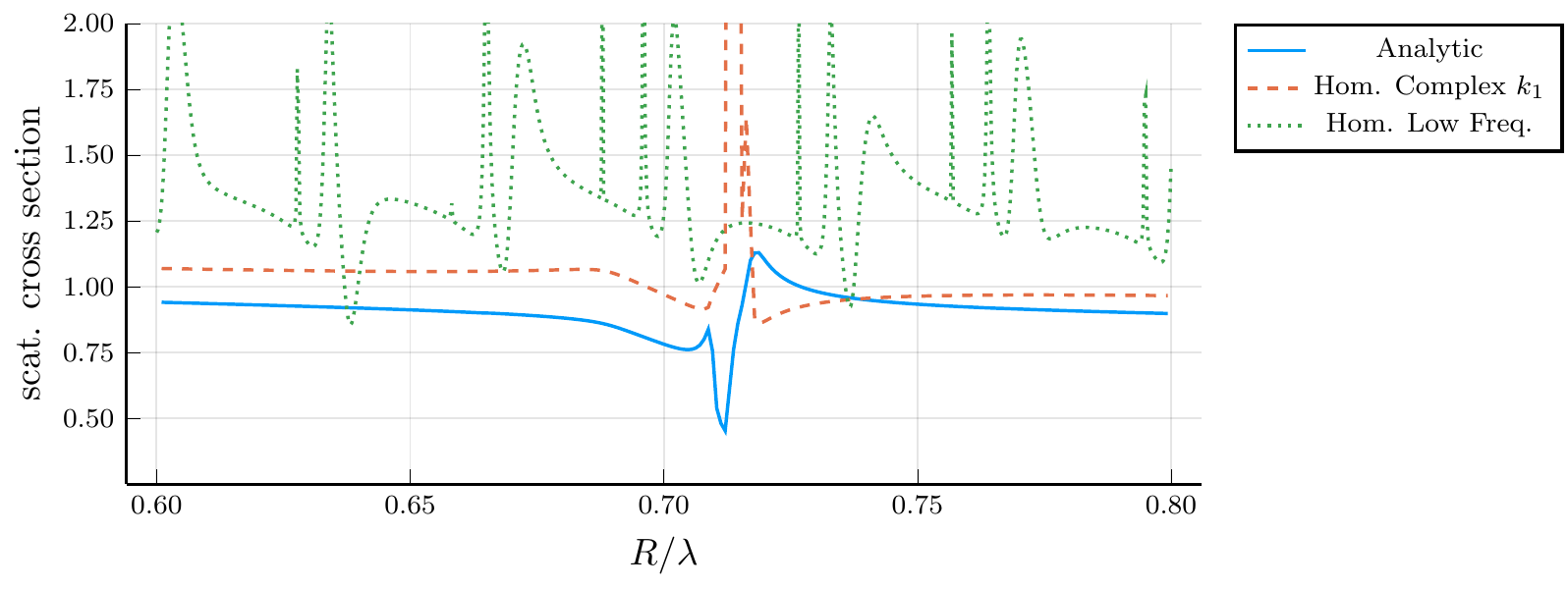}
    \caption{Shows the scattering cross-section of the average scattered field from a sphere filled with particles. The three different methods used are described in the beginning of this section, and the properties used are given in Table~\ref{tab:numerical-properties}. The Hom. Complex $k_1$ exhibits a strong resonance, with the peak climbing above 13, which is not shown to avoid zooming out too far. At a similar frequency, the Analytic solution exhibits the opposite, where scattering is very weak. Note we did not show the range $0.25< R/\lambda < 0.6$ as it is less interesting. }
    \label{fig:cross-section}
\end{figure}

{\bf Frequency sweep.} We begin with a frequency sweep and use the particle properties given in Table~\ref{tab:numerical-properties}. For each frequency, we calculate the scattering cross section for the three methods described above: the analytic and the two homogeneous spheres. The results are shown in Figure~\ref{fig:cross-section}.


We define the non-dimensional scattering cross section by~\cite{Varadan+Varadan1986}
\begin{equation*}
    \text{(scat. cross section)} =
    \frac{1}{2\pi(kR)^2}\sum_n|F_n|^2,
\end{equation*}
where the $|\cdot|$ represents the absolute value. The above is dimensionless and the natural way to compare with the geometrical cross section of the sphere~\cite{Kristensson2016}. In the standard notation, $\sigma_s$ often denotes the scattering cross section, in which case our non-dimensional scattering cross section is equal to $\sigma_s/ (2 \pi R^2)$.


As expected the three methods converge for low frequencies, as shown in Figure~\ref{fig:cross-section}. For $R/\lambda > 0.05$ the Homogeneous Low Frequency sphere quickly diverges from the other two solutions, and then has far more resonant frequencies. The two methods that use the same effective wavenumber, $k_1$, stay closer together, but are significantly different even before reaching $R/\lambda = 0.6$. Around $R/\lambda = 0.72$ we see that both the analytic and the Homogeneous Complex $k_1$ methods hit a resonant frequency, but display very different responses. To further investigate this, we plot the full scattered field for the three methods in Figure~\ref{fig:scattered-field-abs}. These fields show contour maps for the slice $y =0$. The main difference between the methods is that the analytic solution has a weaker scattered field and also has a smaller shadow region.
\begin{figure}[ht]
    \centering
    \includegraphics[height = 6cm]{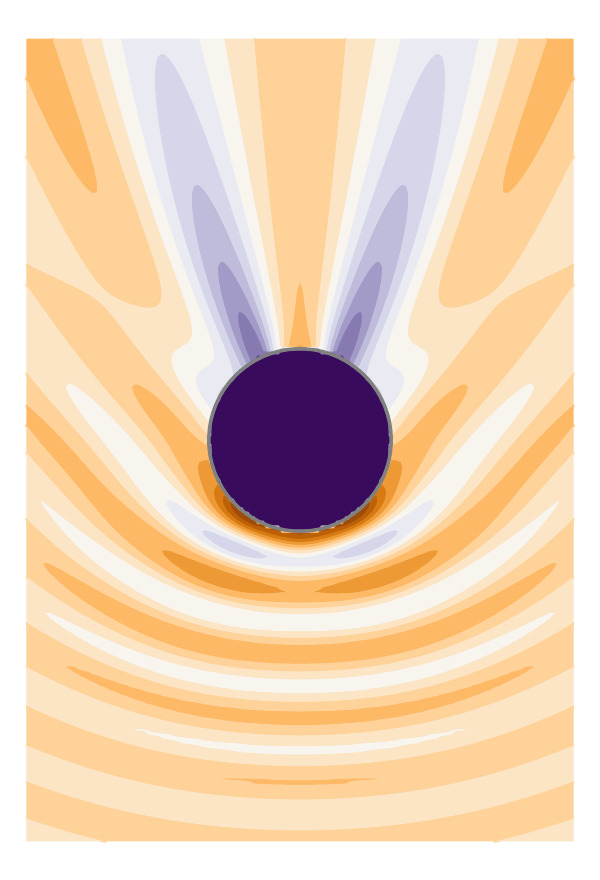}
    \includegraphics[height = 6cm]{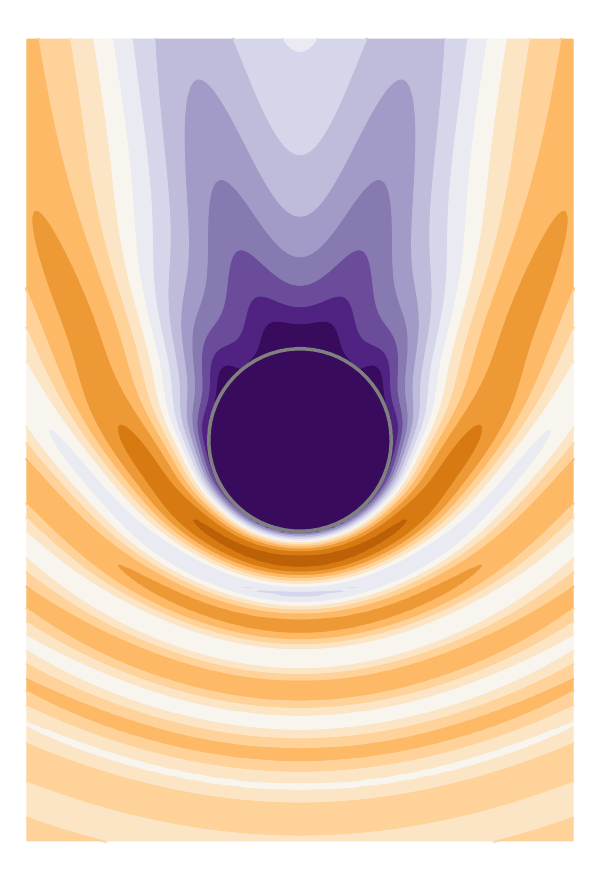}
    \includegraphics[height = 5.98cm]{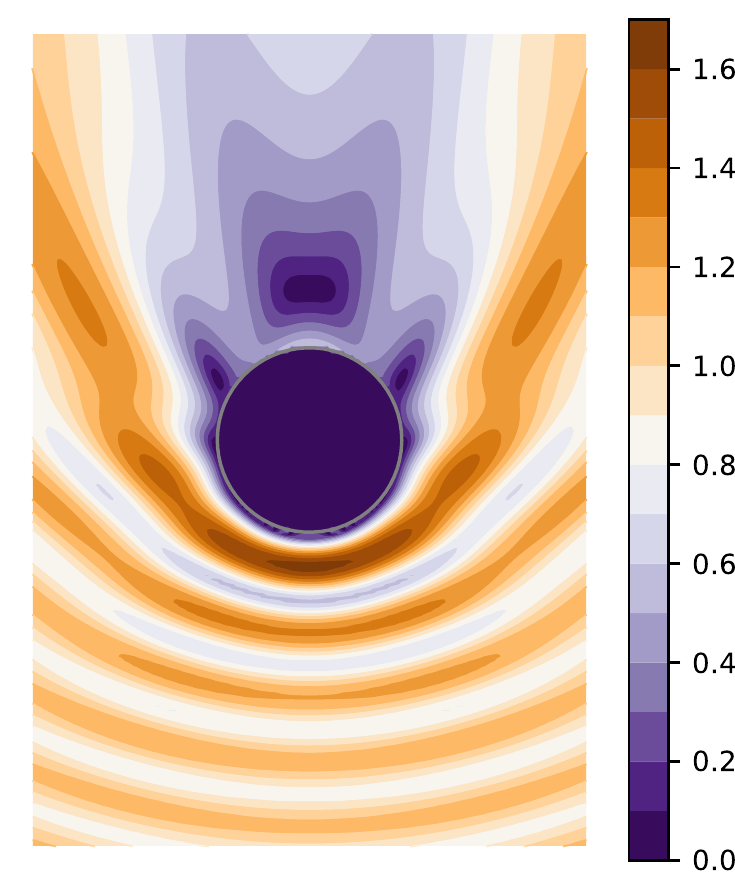}
     \caption{ Comparing the absolute value of the scattered field of the three methods (from left to right): analytic, hom. complex $k_1$, and hom. low freq. for $R/\lambda = 0.71$. Figure~\ref{fig:cross-section} shows the scattering cross-section for these three methods over a large frequency range.}
    \label{fig:scattered-field-abs}
\end{figure}

\begin{figure}[ht]
    \centering
    \includegraphics[width=0.98\linewidth]{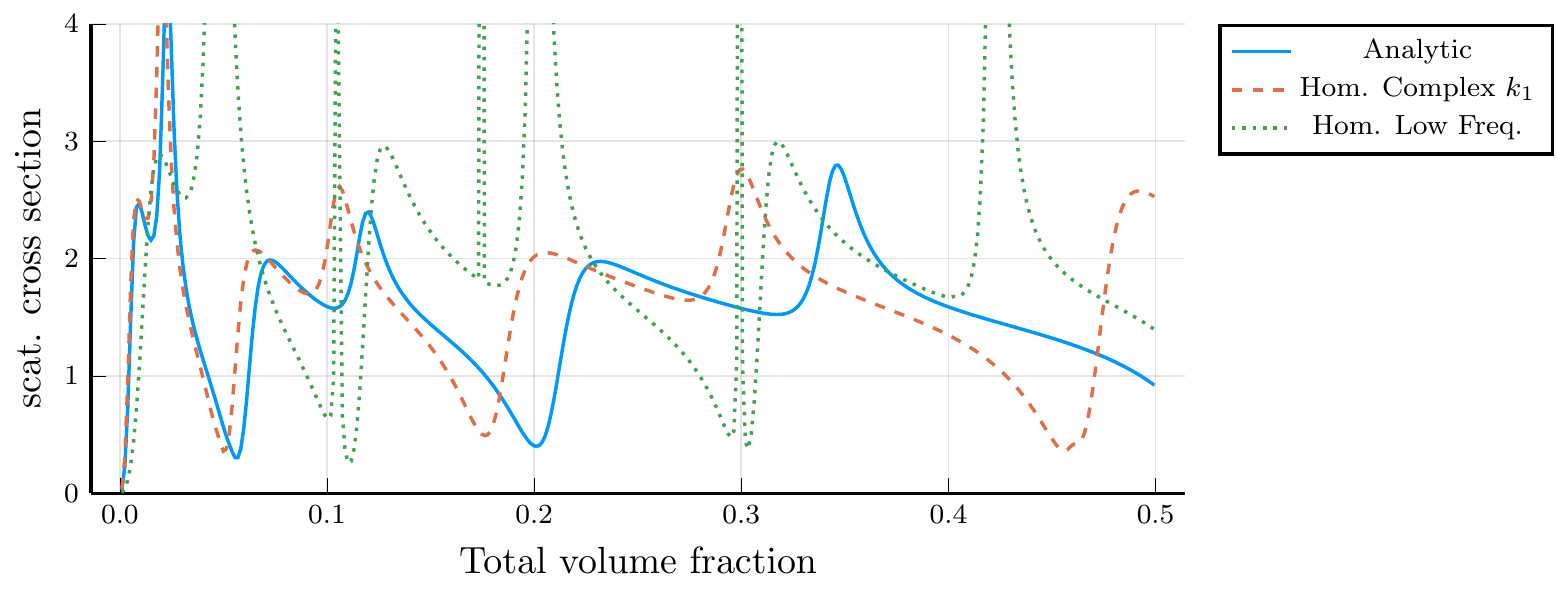}
    \caption{For a fixed frequency which corresponds to $R / \lambda = 0.133$, chosen as it is a lowish frequency, see into, and it is a local min. in Figure~\ref{fig:cross-section}.   }
    \label{fig:cross-section-volfrac}
\end{figure}

{\bf Varying the particle volume fraction.} The effects of multiple scattering between particles {vary} significantly with the volume fraction of the particles as shown in Figure~\ref{fig:cross-section-volfrac}. To produce these results we used a fixed frequency which corresponds to $R / \lambda = 0.133$. This frequency was chosen as it is relatively low and is the local minimum of the analytic scattering cross-section in Figure~\ref{fig:cross-section}. As this is a relatively low frequency, it avoids the need to use multiple effective wavenumbers even for large volume fraction, as described in the beginning of this section.

For a moderate volume fraction, Hom. Complex $k_1$ is qualitatively a good approximation (except close to resonant frequencies) as shown in Figure~\ref{fig:cross-section}. However, when increasing the volume fraction we see a clear drift between Hom. Complex $k_1$ and the Analytic method in Figure~\ref{fig:cross-section-volfrac}. Again we notice that Hom. Low Freq. has more resonant frequencies, and they are more extreme.

\section{Discussion}
Much has already been understood about a plate, or half-space, filled with a random mix of particles, including how to calculate, and make sense of, the effective wavenumbers, reflection, and transmission~\cite{willis2020transmission-meta,willis2020transmission,caleap2012effective,gower2019multiple,gower2019proof,linton_multiple_2006,Tsang+Kong2001,mishchenko_first-principles_2016}. These results are now used to probe emulsions, colloids, and slurry~\cite{challis_ultrasound_2005,quote-bs_iso_20998-3} with sound, and planetary systems with light~\cite{mishchenko_first-principles_2016}, among other applications.

A question that remained was: how to make sense of other regions $\reg$ not shaped like a plate? For example, like a droplet filled with a particulate.

{\bf Effective wavenumbers.} One milestone of this paper was to show that any region $\reg$ filled with the same particulate material will have the same effective wavenumbers, and these effective wavenumbers are given by solving~\eqref{eqn:planewave_det} or \eqref{eqn:azi-planewave_det}, with the low frequency properties given by~\eqref{eqns:effective-properties}. In fact, \eqref{eqn:azi-planewave_det} is simpler to solve then the dispersion equations previously presented in the literature.

A key step we used was to  represent the average wave as a sum of wave potentials each with a different effective wavenumber, as shown by~\eqref{eqn:effective_wave_representation}. This representation is useful because the sum converges. This happens because most of the wave potentials decay rapidly due to their wavenumbers having a large imaginary part as shown in Figures~\ref{fig:wavenumbers-compare} and \ref{fig:weak-wavenumbers}. For more details see~\cite{gower2019proof}.

{\bf Multiple effective wavenumbers.} In this work, we concentrated on scenarios where all the effective wavenumbers, except one, lead to wave modes which decay rapidly. That is, we make use of only one effective wavenumber. This scenario, which occurs for most frequencies and particle properties simplifies the equations. See~\cite{gower2019multiple} for an example where many effective wavenumbers are used. It remains an open challenge to find a simply way to incorporate all effective wavenumbers for scenarios such as a sphere filled with particles.

{\bf The ensemble wave equation.} Our results have enabled us to take effective wavenumbers from a halfspace, or a plate, and use them to calculate the average scattered wave from a sphere filled with particles. To our knowledge, we are the first to provide a clear first-principals approach to achieve this.  Beyond the examples we present in this paper, like a sphere filled with particles, our ensemble wave equation~\eqref{eqn:ensemble_wave} and ensemble boundary conditions~\eqref{eqn:general_boundary} can be used to calculate the average field for regions of any shape. Though, depending on the shape, this may require considerable work.

{\bf Numerical results}. To both demonstrate that our method can completely describe the average scattered field, and to compare  with previous approaches, we present some numerical results for a sphere filled with particles in Section~\ref{sec:numerical-results}. We compared our method with approaches which assume the region $\reg$ is made of some homogeneous material with effective properties. As expected, the different methods converge for low-frequency, as shown in Figures~\ref{fig:cross-section} and \ref{fig:cross-section-volfrac}, though there are significant differences for finite frequencies. For one specific frequency, the difference between the methods is illustrated by a field plot in Figure~\ref{fig:scattered-field-abs}.


{\bf Validation.} The next natural step is to validate our models. Numerical validation would be ideal, as there are robust numerical methods for multiple scattering~\cite{ganesh2015efficient,mishchenko_first-principles_2016,Mackowski2001}. Numerical methods can also clarify the assumptions used in the modelling, such as the choice of pair-correlation and the quasi-crystalline approximation. However, a major issue, that has prevented substantial validation, is that these numerical methods have struggled to simulate an infinite halfspace or infinite  plate required by most of the available theoretical predictions~\cite{gower_backscattering_2018,chekroun_multiple_2012,chekroun_comparison_2009}.

Now, with our framework, numerical validation for finite sized sphere filled with particles should be straight-forward. This will allow a clear way to verify the statistical assumptions used, and the range of their validity. 

{\bf Electromagnetism and Elastodynamics.} Our framework deals with the scalar wave equation. There exist in the literature clear routes on how to extend effective wave theory from the scalar version to elastodynamics~\cite{conoir_effective_2010}, thermo-visco-elasticity~\cite{luppe2012effective,pinfield2014thermo} and electromagnetism~\cite{Tishkovets+etal2011,doicu2019electromagnetic-I,Doicu+Mishchenko2019a}, though each requires extra algebraic manipulation. In light scattering, it is far easier to measure the average of the scattered intensity~\cite{mishchenko_first-principles_2016}, though it requires the average of the scattered field, which is what we calculate in this work. Extending our framework to calculate the average intensity should enable accurate models for scattering from spheres and other compact objects.

\section{Acknowledgements}
Gerhard wishes to gratefully thank the UK Acoustic Network funded by EPSRC (EP/R005001/1) for a generous travel support which made it possible for Gerhard to visit Sheffield in the fall 2019. The authors would also like to acknowledge the late Michael Mishchenko for putting the authors in touch, which ultimately led to this paper, and for his amazing contribution to the field of scattering. The authors are also thankful to Thomas Wriedt for organising the Bremen Workshop on Light Scattering.

\appendix

\section{Spherical harmonics}\label{sec:Functions}
The associated Legendre functions, defined for non-negative integers $\ell\geq m\geq0$, are denoted $\mathrm{P}_\ell^m(x)$, and defined by
\begin{equation*}
\mathrm{P}_\ell^m(x)=(1-x^2)^{m/2}\frac{\diff^m}{\diff x^m}\mathrm{P}_\ell(x),\quad x\in[-1,1],
\end{equation*}
where $\mathrm{P}_\ell(x)$ is the Legendre polynomials.
For a negative integer value of $m$, we use ($m=-1,-2,\ldots$)
\begin{equation*}
\mathrm{P}_\ell^{m}(x)=(-1)^{|m|}\frac{(\ell-|m|)!}{(\ell+|m|)!}\mathrm{P}_\ell^{|m|}(x),\quad x\in[-1,1].
\end{equation*}

The spherical harmonics are denoted $\mathrm{Y}_{\ell m}(\theta,\phi)$ and they are defined by~\cite[(2.5.29), p.~24]{Edmonds1974}
\begin{equation} \label{eqn:spherical-harmonics}
	\mathrm{Y}_{n}(\rvh)=
    \mathrm{Y}_{\ell m}(\theta,\phi)=(-1)^m
	\sqrt{\frac{2\ell+1}{4\pi}\frac{(\ell-m)!}{(\ell+m)!}}\mathrm{P}_\ell^{m}(\cos\theta)\eu^{\iu m\phi},
\end{equation}
where we committed a small abuse in notation as $\rvh = (\cos \phi \sin \theta,\sin \phi \sin \theta,\cos \theta)$, where the angles $\theta$ and $\phi$ can be complex. The indices $\ell$ and $m$ take the following values:
\begin{equation*}
    m=-\ell,-\ell+1,\ldots,-1,0,1,\ldots,\ell,\quad \ell=0,1,2,\ldots.
\end{equation*}
For the special case $\kvhi=\zvh$ we have that
\begin{equation*}
4\pi\mathrm{Y}_{n}^*(\zvh)=\delta_{m,0}\sqrt{4\pi(2\ell+1)}.
\end{equation*}

The spherical harmonics satisfy the parity relation and complex conjugate
\begin{equation*}
\mathrm{Y}_{n}(-\rvh)=(-1)^\ell\mathrm{Y}_{n}(\rvh), \qquad
\mathrm{Y}_{\ell m}^*(\rvh^{*})=(-1)^m\mathrm{Y}_{\ell-m}(\rvh),
\end{equation*}
and $\mathrm{Y}_{n}(\rvh)$ are orthonormal over the real unit sphere $\Omega$, that is
\begin{equation*}
    \int_\Omega \mathrm{Y}_{\ell m}^*(\theta,\phi)\mathrm{Y}_{\ell' m'}(\theta,\phi)\sin\theta\,\diff\theta\diff\phi=(-1)^m\int_\Omega \mathrm{Y}_{\ell -m}(\theta,\phi)\mathrm{Y}_{\ell' m'}(\theta,\phi)\sin\theta\,\diff\theta\diff\phi=\delta_{\ell,\ell'}\delta_{m,m'}.
\end{equation*}

Plane waves can be expanded in terms of spherical harmonics by using:
\begin{equation} \label{eqn:plane-wave-expansion}
    \eu^{\iu \vec x \cdot \vec y} = 4 \pi \sum_{n_1} \iu^{\ell_1}(-1)^{m_1} j_{\ell_1}(x y) \mathrm Y_{n_1}(\unitvec{x}) \mathrm Y_{\ell_1 -m_1}(\unitvec{y}) = 4 \pi \sum_{n_1} \iu^{\ell_1}(-1)^{m_1} \mathrm v_{n_1} (y \vec x) \mathrm Y_{\ell_1 -m_1}(\unitvec{y})
\end{equation}
where both $\vec x$ and $\vec y$ can be complex vectors, and we use the dot product to mean $(x_1,x_2,x_3) \cdot (y_1,y_2,y_3) =  x_1 y_1 + x_2 y_2 + x_3 y_3$ with no conjugation.

\section{Translation matrices}\label{sec:Translation}

The translation properties of the spherical waves are instrumental for the formulation and the solution of the scattering problem of many individual particles.
These translation properties are well know, and we refer to, \eg~\cite{Bostrom+Kristensson+Strom1991,Friedman+Russek1954} for details.
Some of their properties are reviewed in this appendix and a simple proof of these matrices are given in the supplementary material.

Let $\rv'=\rv+\dv$, then
the translation matrices for a translation $\dv$ are~\cite{Bostrom+Kristensson+Strom1991}
  \begin{equation}\label{eq:translation_spherical_waves}
 \left\{\begin{aligned}
   &\mathrm{v}_n(k\rv')=\sum_{n'}\mathcal{V}_{nn'}(k\dv)\mathrm{v}_{n'}(k\rv),\quad \text{ for all }\dv
   \\
 &\mathrm{u}_n(k\rv')=\sum_{n'}\mathcal{V}_{nn'}(k\dv)\mathrm{u}_{n'}(k\rv),\quad |\rv|>|\dv|
 \\
   &\mathrm{u}_n(k\rv')=\sum_{n'}\mathcal{U}_{nn'}(k\dv)\mathrm{v}_{n'}(k\rv),\quad |\rv|<|\dv|
   \\
 \end{aligned}\right..
 \end{equation}

Translation in the opposite direction is identical to the Hermitian conjugate of the translation matrices~\cite{Peterson+Strom1973}, \ie
\begin{equation}
     \mathcal{V}_{nn'}(-k\dv)=\mathcal{V}_{n'n}^*(k^*\dv)=(-1)^{\ell-\ell'}\mathcal{V}_{nn'}(k\dv),\quad \mathcal{U}_{nn'}(-k\dv)=(-1)^{\ell-\ell'}\mathcal{U}_{nn'}(k\dv).
    \label{eqn:translation_transpose}
\end{equation}
The translation matrix $\mathcal{V}_{nn'}(k\dv)$ is identical to $\mathcal{U}_{nn'}(k\dv)$ but with $\mathrm{h}_{\lambda}^{(1)}(k|\dv|)$ replaced with $\mathrm{j}_{\lambda}(k|\dv|)$.

Notice that the translation matrices $\mathcal{V}_{nn'}(k\dv)$ and $\mathcal{U}_{nn'}(k\dv)$ have the form
 \begin{equation}\label{eq:Translation matrix-scalar}
 \mathcal{V}_{nn'}(k\dv)=\sum_{n_1}c_{n n'n_1}{\mathrm v}_{n_1}(k\dv),\quad
 \mathcal{U}_{nn'}(k\dv)=\sum_{n_1}c_{n n'n_1}{\mathrm u}_{n_1}(k\dv).
 \end{equation}
 where the summation over the multi-index $n_1=\{\ell_1, m_1\}$ effectively is over $|\ell-\ell'|\leq \ell_1\leq \ell+\ell'$, and $m_1=m-m'$.
 The explicit values of the coefficients $c_{n n'n_1}$ are, see the supplementary material
\begin{equation} \label{eqn:c-definition}
    c_{n n'n_1}=4\pi\iu^{\ell'-\ell+\ell_1}\int\limits_{\Omega}\mathrm{Y}_{n}(\theta,\phi)\mathrm{Y}_{n'}^*(\theta,\phi) \mathrm{Y}_{n_1}^*(\theta,\phi)\,\sin\theta\diff\theta\diff\phi.
\end{equation}
which can be expressed with the Wigner 3-$j$ symbol~\cite[(4.6.3), p.~63]{Edmonds1974} in the form
\begin{equation} \label{eqn:c-3j}
c_{n n'n''}=\iu^{\ell'-\ell+\ell''}(-1)^{m}
   \sqrt{4\pi(2\ell+1)(2\ell'+1)(2\ell''+1)}
  \begin{pmatrix}
      \ell&\ell'&\ell''
      \\
      0&0&0
  \end{pmatrix}
  \begin{pmatrix}
      \ell&\ell'&\ell''
      \\
      m&-m'&-m''
  \end{pmatrix}.
\end{equation}
Note that the coefficients $c_{nn'n''}$ are all real due to orthogonality in the azimuthal index. Further the $c_{nn'n''}$ are only non-zero when
\begin{equation}\label{eqns:gaunt_indices}
m-m'=m'', \quad |\ell - \ell'| \leq \ell'' \leq \ell + \ell', \quad \ell + \ell' + \ell '' = \text{even integer},
\end{equation}
and should only be evaluated for $\ell, \ell',\ell'' \geq 0$ and
\begin{align}
-\ell \leq m \leq \ell, \quad -\ell' \leq m' \leq \ell', \quad -\ell'' \leq m'' \leq \ell''.
\end{align}

Other often used notation is the Gaunt coefficient~\cite{martin_multiple_2006}:
\begin{equation} \label{eqn:Gaunt-coefficient}
4\pi {\cal G}(\ell,m;\ell'm';\ell_1)=\iu^{-\ell'+\ell-\ell_1} (-1)^{m'}c_{(\ell,m) (\ell',-m')(\ell_1,m+m')},
\end{equation}
and the Clebsch-Gordan coefficients
\begin{equation}
    c_{n n'n''} = \iu^{\ell'-\ell+\ell''}(-1)^{m'} \sqrt{4\pi \frac{(2\ell+1)(2 \ell'+1)}{(2 \ell''+1)}}
    \langle \ell \,0 \,\ell' \,0 | \ell'' \,0 \rangle
    \langle \ell \,m \,\ell'\, -m' | \ell''\, m'' \rangle.
\end{equation}

The special case $c_{n n (0,0)} = \sqrt{4 \pi}$ and following properties are useful:
\begin{align}  \label{eqn:c-transpose}
    & c_{n n'n''} = c_{n n'' n'} = c_{(\ell,-m)(\ell',-m')(\ell'',-m'')},
    \\
    & c_{n n'n''} = (-1)^{m''+\ell''} c_{n'n(\ell'',-m'')}  = (-1)^{\ell' +m'} c_{n''(\ell',-m')n},
    \\
    \label{eqn:spherical-linearisation}
    & \sum_{n_1} \iu^{\ell_1}\mathrm Y_{n_1}(\theta,\phi) c_{n_1 n' n} = 4 \pi \iu^{\ell+\ell'} \mathrm Y_{n}(\theta,\phi) \mathrm Y_{n'}(\theta,\phi),
\end{align}
where the last is the contraction rule, or the linearisation formula~\cite{martin_multiple_2006}. For real $\theta$ and $\phi$ the linearisation formula can be deduced by multiplying both sides of~\eqref{eqn:spherical-linearisation} by $\mathrm Y_{n_2}^*(\rvh)$, then integrating over $\rvh$, and applying the definition~\eqref{eqn:c-definition}.

\section{Separating the effective waves in equation~\eqref{eqn:combined_green_waves}}

In this appendix, we address the solution of an equation of the form
\begin{equation*}
    \sum_{p=0}^P\psi_p(\rv)=0,
\end{equation*}
where the functions $\psi_p(\rv)$ satisfy
\begin{equation*}
    \nabla^2\psi_p(\rv)=-k_p^2\psi_p(\rv),
\end{equation*}
where $k_p\neq k_q$, $p\neq q$.
The following theorem proves that the solution of this equation is $\psi_p(\rv)=0$, $p=0,1,2,\ldots,P$:
\begin{theorem}\label{th:Vandermonde}
    Let the functions $\psi_p(\rv)$ for $p=0,1,2,\ldots,P$, satisfy
$\nabla^2\psi_p(\rv)=\alpha_p\psi_p(\rv)$ for $\rv \in \reg$. Assuming $\alpha_p\neq\alpha_q$ for every $p\neq q$, then the only solution to
\begin{equation} \label{eqn:sum-waves}
    \sum_{p=0}^P\psi_p(\rv)=0, \quad \text{for all} \;\; \rv \in \reg,
\end{equation}
is $\psi_p(\rv)=0$ for $p=0,1,2,\ldots,P$.
\end{theorem}
\begin{proof}
From the assumption in the theorem, we have
\begin{equation*}
   \nabla^2\sum_{p=0}^P\psi_p(\rv)=\sum_{p=0}^P \alpha_p\psi_p(\rv)=0,
\end{equation*}
for any open ball within $\reg$.
Or, more generally, by repeated use of the Laplace operator
\begin{equation*}
   \sum_{p=0}^P \alpha_p^n\psi_p(\rv)=0, \quad \text{for}\;\; n = 0,1,2,\ldots,P,
\end{equation*}
which we summarise in a matrix notation
\begin{equation*}
    \begin{pmatrix}
    1&1&1&\cdots&1\\
    \alpha_0&\alpha_1&\alpha_2&\cdots&\alpha_P\\
    \alpha_0^2&\alpha_1^2&\alpha_2^2&\cdots&\alpha_P^2\\
    \alpha_0^3&\alpha_1^3&\alpha_2^3&\cdots&\alpha_P^3\\
    \vdots&\vdots&\vdots&\ddots&\vdots\\
    \alpha_0^P&\alpha_1^P&\alpha_2^P&\cdots&\alpha_P^P
    \end{pmatrix}
    \begin{pmatrix}\psi_0(\rv)\\\psi_1(\rv)\\\psi_2(\rv)\\\psi_3(\rv)\\\vdots\\
    \psi_P(\rv)
    \end{pmatrix}=\begin{pmatrix}0\\0\\0\\0\\\vdots\\0\end{pmatrix}.
\end{equation*}
The matrix on the left-hand side is the transpose of the  Vandermonde matrix with determinant $\prod_{0\leq p<q\leq P}(\alpha_q-\alpha_p)$, see \eg~\cite{Davis1975}.
Under the assumption $\alpha_p\neq\alpha_q$, $p\neq q=0,1,2,\ldots,P$, this matrix determinant is non-zero, and we obtain the result of the theorem, $\psi_p(\rv)=0$ for $p=0,1,2,\ldots,P$ and $\rv$ within some open ball in $\reg$. {By analyticity of the solutions to the Helmholtz equation, the functions $\psi_p(\rv)=0$ for $p=0,1,2,\ldots,P$ and $\rv\in\reg$.}
\end{proof}

\section{The matrix $G_{n,n_2}$}\label{sec:Wronskian}
Here, we calculate the dimensionless matrix $G_{n,n_2}(\lambda_1)$ by substituting the spherical basis expansion~\eqref{eqns:fields_spherical_basis} into~\eqref{eqn:Ib}, which leads to
\begin{equation*}
    \mathcal J_p(\rv_1) = \sum_{n_1 n_2} F_{p,n'n_1}(\lambda_2) \mathrm v_{n_2}(k_p \rv_1) \int_{\partial B(\vec 0;a_{12})}  \Biggl\{\mathcal{U}_{n'n}(-k\rv) \frac{\partial \mathcal{V}_{n_1n_2}(k_p\rv)}{\partial \vec\nu}
     - \frac{\partial\mathcal{U}_{n'n}(-k\rv)}{\partial \vec\nu} \mathcal{V}_{n_1n_2}(k_p\rv)\Biggr\}\diff A,
\end{equation*}
From Appendix~\ref{sec:Translation}, we have that
\begin{align*}
    & \mathcal{U}_{n'n}(-k\rv) = (-1)^{\ell'-\ell}\sum_{n''}c_{n' nn''}\mathrm{u}_{n''}(k\rv),
    \\
    &\mathcal{V}_{n_1n_2}(k_p\rv) = \sum_{n_3}c_{n_1 n_2n_3}\mathrm{v}_{n_3}(k_p\rv).
\end{align*}

Integrating over the spherical surface $\partial B(\vec 0;a_{12})$, and using the orthogonality of the spherical harmonics, we obtain
\begin{equation}
    \mathcal J_p(\rv_1) = - \sum_{n_1 n_2} F_{p,n'n_1}(\lambda_2) \mathrm v_{n_2}(k_p \rv_1) \sum_{n_3}
c_{n n'n_3} c_{n_1 n_2n_3} a_{12} \mathrm {N}_{\ell_3}(ka_{12},k_p a_{12}),
\end{equation}
where $\mathrm N_{\ell}(x,z)$ is defined by \eqref{eqn:N_ell}.
When substituting the above into \eqref{eqn:ensemble_wave} leads to the matrix $G_{n,n_2}(\lambda_1)$ defined by \eqref{eqn:G-simple}.

\section{Effective plane-waves}\label{sec:Effective plane waves}
Here, we show that we recover the plane-wave dispersion equation deduced in much of the literature from our general ensemble wave equation~\eqref{eqn:ensemble_wave}.
We use the plane-wave representation~\eqref{eqn:effective-plane-wave} together with~\eqref{eqn:plane-wave-expansion} to write
\begin{align}
  & f_{p,n'}(\rv_1,\lambda_1) = F_{p,n'}(\lambda_1)\eu^{\iu \vec k_p \cdot \rv_1} = 4 \pi F_{p,n'}(\lambda_1) \sum_{n_1} \iu^{\ell_1} \mathrm j_{\ell_1}(k_p r_1)  Y_{n_1}^*(\unitvec{r}_1) Y_{n_1}(\unitvec{k}_p),
  \\
  & f_{p,n'}(\rv_1 + \rv,\lambda_2) = 4 \pi f_{p,n'}(\rv_1,\lambda_2)  \sum_{n_1} \iu^{\ell_1} \mathrm j_{\ell_1}(k_p r)  Y_{n_1}^*(\unitvec{r}) Y_{n_1}(\unitvec{k}_p),
\end{align}
where we used~\eqref{eqn:Complex angles}.

Using the above, we can simplify~\eqref{eqn:ensemble_wave} by calculating:
\begin{gather*}
    \mathcal J_p(\rv_1) =
     f_{p,n'}(\rv_1,\lambda_2)  4 \pi \sum_{n_1}   \iu^{\ell_1} Y_{n_1}(\unitvec{k}_p)  \int_{\partial B(\vec 0;a_{12})}   Y_{n_1}^*(\unitvec{r})\left( \mathcal{U}_{n'n}(-k\rv) \frac{\partial  \mathrm j_{\ell_1}(k_p r) }{\partial r}
     - \frac{\partial \mathcal{U}_{n'n}(-k \rv) }{\partial r}  \mathrm j_{\ell_1}(k_p r)  \right)\mathrm dA
     \\
     = - f_{p,n'}(\rv_1,\lambda_2)  4 \pi  a_{12}  (-1)^{\ell'-\ell}  \sum_{n_1}  \iu^{\ell_1} Y_{n_1}(\unitvec{k}_p) c_{n' nn_1} \mathrm{N}_{\ell_1} (k a_{12}, k_p a_{12}),
\end{gather*}
where we used $\mathcal{U}_{nn'}(-k\dv)=(-1)^{l-l'}\mathcal{U}_{nn'}(k\dv)$
followed by $\mathcal{U}_{nn'}(k\dv)=\sum_{n_1}c_{n n'n_1}{\mathrm u}_{n_1}(k\dv)$, and~\eqref{eqn:N_ell}.
Substituting the above into~\eqref{eqn:ensemble_wave} then leads to the plane-wave eigensystem~\eqref{eqn:planewave_eigensystem}.
The above dispersion equation is the same\footnote{After making the substitutions $T_n \to -Z_\ell$, $a_{12} \to b_{12}$, $c_{n'nn_1} \to \delta_{m_1,m'-m} 4 \pi (-1)^{m} \iu^{\ell_1 + \ell - \ell'} \mathcal G(\ell,m;\ell',-m';\ell_1)$, $F_{p,n}  \to 4\pi \iu^\ell  Z_\ell F_\ell^m$, and $ N_q (k a_{12}, k_p  a_{12}) \to -N_q (k_p b_{12}) (\mathcal \iu k b_{12})^{-1}$, followed by using  $(-1)^{\ell ' + \ell + \ell_1} = 1$. Note that our conventions of spherical harmonics is the same theirs, and that $F_\ell^m$ does not depend on $p$ as multiple effective waves was not considered in \cite{linton_multiple_2006}.
}
as \cite[equation (4.20)]{linton_multiple_2006} (where you need to set $\mathcal A_n^m = 0$) when considering a single species (no integer over $\mathcal S$), and only one effective wave.

\section{Integrals of spherical and plane waves}

When dealing with effective plane-waves, we need to evaluate the following integral:
\begin{equation} \label{def:L}
L_{n}(z) =  \int_{\R^2}  \mathrm{u}_{n}(k\rv) \frac{\partial \eu^{\iu\vec{k}_p\cdot\rv}}{\partial z} - \frac{\partial \mathrm{u}_{n}(k\rv) }{\partial z} \eu^{\iu\vec{k}_p\cdot\rv} \,\diff x\diff y.
\end{equation}
These integrals converge when Im $k \geq |\mathrm{Im}\, {k_p}_x| + |\mathrm{Im}\, {k_p}_y|$, where $\vec k_p = ({k_p}_x,{k_p}_y,{k_p}_z)$. This inequality holds when using planar symmetry~\eqref{eqn:plane-wave-symmetry}, which together with $\vec k = (k_x,k_y,k_z)$ implies that ${k_p}_x = k_x$ and ${k_p}_y = k_y$.

To calculate $L_{n}(z)$ we employ a transformation~\cite{Danos+Maximon1965,Bostrom+Kristensson+Strom1991,Kristensson2016,Kristensson1979a} between radiating spherical waves and plane waves:
\begin{equation}
    \mathrm{u}_{n}(k\rv)=\frac{1}{2\pi\iu^{\ell}}\int_{\R^2}\mathrm{Y}_{n}(\unitvec{q})
    \eu^{\iu \vec q \cdot \rv }
    \frac{\diff q_x\diff q_y}{k q_z},\quad \text{for} \;\; z>0,\;\; \Re k > 0, \;\; \Im k \geq 0
\end{equation}
where $\vec q = (q_x,q_y,q_z)$,
$q_z=(k^2-q_x^2-q_y^2)^{1/2}$ and evaluated such that
$\Im q_z \geq 0$. If $z < 0$ we use $\mathrm{u}_{n}(k\rv) = (-1)^{\ell}\mathrm{u}_{n}(-k\rv)$ and then apply the above.

Substituting the above representation into $L_{n}(z)$ leads to
\begin{equation}
L_{n}(z) = \begin{cases}
    \displaystyle
   \frac{1}{2 \pi \iu^{\ell}} \int_{\R^2} \left( \int_{\R^2}\iu({k_p}_z-q_z)
    \mathrm{Y}_{n}(\unitvec{q})
    \eu^{\iu (\vec k_p + \vec q) \cdot \rv}
    \frac{\diff q_x\diff q_y}{kq_z}\right) \diff x\diff y, & z>0,
    \\
    \displaystyle
    \frac{(-1)^{\ell}}{2 \pi \iu^{\ell}} \int_{\R^2} \left( \int_{\R^2}\iu({k_p}_z + q_z)\mathrm{Y}_{n}(\unitvec{q})
    \eu^{\iu (\vec k_p - \vec q) \cdot \rv}
    \frac{\diff q_x\diff q_y}{k q_z}\right) \diff x\diff y, & z < 0.
\end{cases}
\end{equation}
Changing the order of integration, then integrating in $x$ and $y$, leads to
\begin{equation}
L_{n}(z) = \begin{cases}
    \displaystyle
   \frac{2 \pi}{ \iu^{\ell}}  \int_{\R^2}\iu({k_p}_z-q_z)\mathrm{Y}_{n}(\unitvec{q}) \delta({k_p}_x + q_x) \delta({k_p}_y + q_y)
    \eu^{\iu ({k_p}_z + q_z) z}
    \frac{\diff q_x\diff q_y}{kq_z}, & z>0,
    \\
    \displaystyle
    \frac{2 \pi}{ (-\iu)^{\ell}} \int_{\R^2}\iu({k_p}_z + q_z)\mathrm{Y}_{n}(\unitvec{q})
    \delta({k_p}_x - q_x) \delta({k_p}_y - q_y) \eu^{\iu ({k_p}_z - q_z) z}
    \frac{\diff q_x\diff q_y}{kq_z}, & z < 0,
\end{cases}
\end{equation}
then integrating in $q_x$ and $q_y$ we get
\begin{equation} \label{eqn:L-integrated}
L_{n}(z) =
\mathrm{Y}_{n}(\hat {\vec k}_p^+) \frac{2 \pi\iu}{ \iu^{\ell}kq_z}
\begin{cases}
    \displaystyle
   (-1)^{m}  ({k_p}_z-q_z)
    \eu^{\iu ({k_p}_z + q_z) z} , & z>0,
    \\
    \displaystyle
    (-1)^{\ell} ({k_p}_z + q_z)
    \eu^{\iu ({k_p}_z - q_z) z} , & z < 0,
\end{cases}
\end{equation}
where $\vec k_p^+ = ({k_p}_x,{k_p}_y,q_z)$ and $q_z = (k^2 - {k_p}_x^2 - {k_p}_y^2)^{1/2}$, and we used that $\mathrm{Y}_{n}(\hat {\vec k}_p^-)  = (-1)^{m}\mathrm{Y}_{n}(\hat {\vec k}_p^+)$ where $\vec k_p^- = (-{k_p}_x,-{k_p}_y,q_z)$.

In most cases where we use plane-waves, we will assume the material occupies the region $\reg_1 = \{z>a_1: \rv \in \mathbb R^3\}$. In this case we have that ${k_p}_x = k_x$ and ${k_p}_y = k_y$, due to planar symmetry~\eqref{eqn:plane-wave-symmetry}, which implies that $\hat {\vec k}_p^+ = \unitvec k$ and $q_z = (k^2 - k_x^2 - k_y^2)^{1/2} = k_z$. Substituting these results in \eqref{eqn:L-integrated} then leads to
\begin{equation} \label{eqn:L-integrated-symmetry}
L_{n}(z) =
\mathrm{Y}_{n}(\hat {\vec k}) \frac{2 \pi\iu}{ \iu^{\ell}k k_z}
\begin{cases}
    \displaystyle
   (-1)^{m}  ({k_p}_z- k_z)
    \eu^{\iu ({k_p}_z + k_z) z}, & z>0,
    \\
    \displaystyle
    (-1)^{\ell} ({k_p}_z + k_z)
    \eu^{\iu ({k_p}_z - k_z) z}, & z < 0,
\end{cases}
\end{equation}
where $k_{p_z}$ is the $z$ component of $\vec k_p$. The case $z <0$ gives the same result obtained in~\cite[Equation B.5]{linton_multiple_2006}.

\referencelist



\end{document}